\documentclass[10pt,a4paper]{article}   

\usepackage[margin=1.25in]{geometry}
\usepackage{hyperref}
\usepackage[pdftex]{graphicx}
\usepackage{amsmath,amssymb,amsfonts}
\usepackage[textsize=small,textwidth=1.3in]{todonotes}
\usepackage[english]{babel}
\usepackage{enumerate}
\usepackage{multicol}
\usepackage{subcaption}
\usepackage[ruled]{algorithm2e}
\usepackage{graphicx}
\usepackage{subcaption}

\usepackage{algorithmic}

\usepackage[utf8]{inputenc}
\usepackage[margin=1.25in]{geometry}
\usepackage[T1]{fontenc}
\usepackage[english]{babel}
\usepackage{color}
\usepackage[textsize=small]{todonotes}

\usepackage{amsfonts}
\usepackage{amsthm}
\usepackage{enumerate}
\usepackage{cite}
\usepackage{multicol}
\usepackage{multirow}
\usepackage{fullpage}
\usepackage{hyperref}

\usepackage[all,dvips]{xy}
\usepackage{graphicx}
\usepackage{tikz}
\usetikzlibrary{matrix,arrows,decorations.pathmorphing}

\theoremstyle{plain}
\newtheorem{theorem}{Theorem}
\newtheorem{proposition}[theorem]{Proposition}
\newtheorem{corollary}[theorem]{Corollary}
\newtheorem{lemma}[theorem]{Lemma}

\theoremstyle{definition}
\newtheorem{definition}{Definition}

\theoremstyle{remark}
\newtheorem{remark}{Remark}
\newtheorem{example}{Example}

\newtheorem*{keywords}{Keywords}
\newtheorem*{PACS}{Mathematics Subject Classification (2010)}

\newcommand{\eqdef}{\stackrel{\text{def}}{=}}

\newcommand{\N}{\mathbb{N}}

\newcommand{\Z}{\mathbb{Z}}
\newcommand{\Q}{\mathbb{Q}}
\newcommand{\R}{\mathbb{R}}
\newcommand{\F}{\mathbb{F}}

\newcommand{\mC}{\mathcal{C}}

\newcommand{\lii}{[\![}
\newcommand{\blii}{\left[\!\!\left[}
\newcommand{\rii}{]\!]}
\newcommand{\brii}{\right]\!\!\right]}

\usepackage{authblk}

\begin{document}

\title{High dimensional affine codes whose square has a designed minimum distance\thanks{Partially supported by the Spanish Ministry of Economy/FEDER: grants MTM2015-65764-C3-1-P, MTM2015-65764-C3-2-P, MTM2015-69138-REDT, MTM2016-78881-P, MTM2016-80659-P, and RYC-2016-20208 (AEI/FSE/UE), and Junta de CyL (Spain): grant VA166G18.}}
%
%
\author[1]{Ignacio Garc\'ia-Marco}
\author[1]{Irene M\'arquez-Corbella}
\author[2]{Diego Ruano}

\affil[1]{
Departamento de Matem\'aticas, Estad\'istica e I.O., Universidad de La Laguna, 38200 La Laguna, Tenerife, Spain. Email: \href{mailto:iggarcia@ull.es}{iggarcia@ull.es} and 
 \href{mailto:imarquec@ull.es}{imarquec@ull.es} 
}


\affil[2]{IMUVA-Mathematics Research Institute, Universidad de Valladolid, 47011 Valladolid, Spain. Email: \href{mailto:diego.ruano@uva.es}{diego.ruano@uva.es}}
%
%

\maketitle

\begin{abstract}
Given a linear code $\mathcal{C}$, its square code $\mathcal{C}^{(2)}$ is the span of all component-wise products of two elements of $\mathcal{C}$. Motivated by applications in multi-party computation, our purpose with this work is to answer the following question: which families of affine variety codes have simultaneously high dimension $k(\mathcal{C})$ and high minimum distance of $\mathcal{C}^{(2)}$, $d(\mathcal{C}^{(2)})$? More precisely, given a designed minimum distance $d$ we compute an affine variety code $\mathcal{C}$ such that $d(\mathcal{C}^{(2)})\geq d$ and that the dimension of $\mathcal{C}$ is high. The best construction that we propose comes from hyperbolic codes when $d\ge q$ and from weighted Reed-Muller codes otherwise.

\end{abstract}

\keywords{Affine variety codes \and Multi-party computation \and Square codes \and Schur product of codes \and Minkowski sum \and convex set}
\PACS{94B05 \and 94B75}

\section{Introduction}

Multi-party computation studies the case where a group of persons, each holding an input for a function, wants to compute the output of it, without having each individual reveal his or her input to the other parties. Multi-party computation is possible from secret sharing schemes \cite{ivan}, and hence from coding theory. From now on, given a linear code $\mC$, the dimension of $\mC$ will be denoted by $k(\mC)$ and its minimum distance by $d(\mC)$. Moreover, if $\mC$ is a linear code over $\mathbb F_q$ of length $n$, dimension $k$ and minimum distance $d$, we call $[n,k,d]_q$ the parameters of $\mC$.

One of the best known protocols is MiniMac \cite{minimac}, which evaluates boolean circuits, and its successor TinyTable \cite{minimac2}. These methods use a linear code $\mathcal{C}$, which should prevent cheating. The probability that a cheating player is caught depends on the minimum distance of $\mathcal{C}*\mathcal{C}=\mathcal{C}^{(2)}$, the square code of linear code \cite{Hugues:2015}, meaning that a high distance on the square will give a higher security. Simultaneously, it would be beneficial to have a code $\mathcal{C}$ with high rate to reduce the communications cost. Therefore, it is desirable to optimize both parameters: $d(\mC^{(2)})$ and $k(\mC)$.

Although, in this article we are more interested in the application of the schur product to the area of secure multiparty computation, this operation has other applications. For example, component-wise products of linear codes have been used to decode linear codes \cite{pellikaan:1992,pellikaan:1996} where it is shown that a linear code of length $n$ with a $t$-error correcting pair has a decoding algorithm which corrects up to $t$ errors with complexity $\mathcal O \left(n^3\right)$. Moreover, the schur product is also used for cryptanalytic applications against the McEliece cryptosystem \cite{CGVO:2013,COT:2014,CMP:2017,marquez:2014,Wieschebrink:2010}, which rely on two assumptions: the generic decoding is hard on average and it is hard to distinguish the public key (a generator matrix of a code $\mC$ with a certain structure) from a random matrix. For a summary of these applications and some others, see \cite[\S 4]{Hugues:2015}.

These applications show the importance of finding linear codes, where both the code itself and the square have good parameters. Choosing a random linear code, with dimension linear in the length, will, with high probability, give a reasonable minimum distance, however, this does not hold for the square code \cite{cascudo1}. Hence, constructing good square codes is a difficult problem. Nevertheless, good square codes exist, since there exists an asymptotic family of codes with the previous property \cite{lona2}. The best binary construction available in the literature is the one in \cite{cascudo2} obtained from cyclic codes, but their constellation is quite limited. A larger constellation can be found at  \cite{jaron}.

Another family of codes that have been proposed for obtaining codes with a good square are \emph{Reed-Muller codes} \cite{Hugues:2015}. Reed-Muller codes can be understood as affine variety codes when one considers the ideal $I=(0)$, i.e. when one evaluates multivariate polynomials ($m$ variables) at all the points of $\mathbb{F}_q^m$. We will restrict our attention to this case of affine variety codes in this article. One has the footprint bound \cite{footprint} for estimating the minimum distance of this family of codes. The family of \emph{hyperbolic codes} \cite{GH01} was introduced to maximize the dimension of an affine variety code given a designed minimum distance from the footprint bound. In particular, a hyperbolic code has a dimension greater than or equal to a Reed-Muller code with the same minimum distance. Hence, it is natural to consider hyperbolic codes for obtaining codes where both the dimension of the code and the minimum distance of the square code are good.

Given $d \in \mathbb{Z}^+$, in this work we propose a method to obtain an affine variety code $\mC$ satisfying that $d(\mC^{(2)}) \geq d$ and such that $k(\mC)$ is considerably high. Our method receives as input a value $d \in \mathbb{Z}^+$ and starts by considering an affine code $\mC_B$ associated with a set $B\subseteq \mathbb N^m$ such that $d(\mC_B) \geq d$, say for example, a hyperbolic code with minimum distance at least $d$. Then, by means of convexity arguments, we build a set $A \subset \mathbb{N}^m$ such that the Minkowski sum $A+A$ is contained in $B$. The latter condition implies that $d(\mC_A^{(2)}) \geq d(\mC_B) \geq d$.
Remarkably, the best candidate for the set $A\subseteq \mathbb N^m$ is not always the one related with a hyperbolic code. Indeed, when the value of the designed minimum distance $d$ is small enough, $d<q$, we prove that there exist certain \emph{weighted Reed-Muller codes} that outperform hyperbolic codes.

Additionally, if the minimum distance of the dual of $\mathcal{C}$ and $d(\mathcal{C}^{(2)})$ are greater than or equal to $t+2$, then $\mathcal{C}$ can be used to construct a $t$-strongly multiplicative secret sharing scheme (SSS). Such a SSS is enough to construct an information theoretic secure secret sharing scheme if at most $t$ players are corrupted \cite{strong1,strong2,strong3}.  This application shows the importance of finding linear codes where $d(\mathcal{C}^\bot)$ is also high relative to the length of the code, where $\mC^{\bot}$ is the dual code of $\mC$. Although in this work we have not focused in maximizing $d(\mathcal{C}^\bot)$ (this is also the case of other articles in the literature as \cite{cascudo2,jaron}), we note that for the affine variety codes considered in this article, the dual of $\mathcal{C}$ is again an affine variety code that can be easily constructed by \cite[Proposition 1]{fercarlos}. Moreover, its minimum distance can also be estimated using the footprint bound.

\subsection*{Outline of the article}

Section \S\ref{Section2} presents the notation used in the article and review some of the standard facts on affine variety codes, in particular, it provides a detailed exposition of the footprint bound, a lower bound for their minimum distance. We also describe some well known examples of affine variety codes which will be essential throughout the article, as Reed-Muller codes, weighted Reed-Muller codes and hyperbolic codes. We end this section with an original result that indicates, in the case of two variables, when the hyperbolic code has strictly higher dimension than a Reed-Muller code with the same minimum distance. Moreover, one can find in the appendix some results (some of them well known) that show when the footprint bound is sharp. We emphasize that Lemma \ref{FB-1} and Lemma \ref{FB-2} have been used in the proof of some results in the article.

Next, in Section \S\ref{Section3} we look more closely at the operation of Schur product of affine variety codes and its relation with the Minkowski sum. Moreover, we present the key result of the article that allows us to establish a strategy to construct affine variety codes whose square code has a designed minimum distance $d$, this strategy is outlined in Algorithm \ref{Algor-1}. That is, given $d\in \mathbb N$ we construct an affine variety code $\mC$ such that $d(\mC^{(2)})\geq d$.

In section \S\ref{Section4} we will be more ambitious, this section contains the main results of the article. If our goal till this section was to obtain an affine code $\mC$ whose square code has designed minimum distance $d$ i.e. $d(\mC^{(2)})\geq d$, throughout Section \S\ref{Section4} our additional goal is providing a code $\mC$ that has also high minimum distance. It seems natural to expect that a code coming from a hyperbolic code  will be the best candidate for our new goal. We have called this type of codes \emph{half hyperbolic codes} and they have been studied in detail in Section \S\ref{Section4.1}.
Surprisingly, \emph{half hyperbolic codes} are not always the best option. Indeed, we prove in Section \S\ref{Section4.2} that, when the value of the designed minimum distance $d$ is small enough, there exist certain weighted Reed-Muller codes that outperform half hyperbolic codes. That is, when $d$ is small enough, then there are weighted Reed-Muller codes $\mathcal D$ whose square has the same designed minimum distance $d$ than the corresponding  half hyperbolic code $\mC$ (i.e. $d(\mC^{(2)})$, $d(\mathcal D^{(2)})\geq d$) and such that $k(\mathcal D) > k(\mC)$.

\section{Affine variety codes}
\label{Section2}
Let us start this section with a brief summary on affine varieties, polynomials and ideals to set up notation and terminology. For a fuller treatment we refer the reader to \cite{Cox-05,Cox-07}.

Let $k[X_1, \ldots, X_m]$ be a ring of polynomials over a field $k$ and consider a monomial ordering $\succ$ on $k[X_1,\ldots,X_m]$. For a polynomial $f\in k[X_1, \ldots, X_m]$ we denote by $\mathrm{in}(f)$ its {\it leading term} with respect to $\succ$, that is, the largest monomial that occurs in $f$. For any ideal $I\subseteq k[X_1, \ldots, X_m]$ we denote by $\mathrm{in}(I)$ its {\it initial ideal}, which is $\mathrm{in}(I) = \left\langle \mathrm{in}(f) \mid f\in I\right \rangle.$
The \emph{radical ideal} of $I$, denoted $\sqrt{I}$, is the ideal  $\sqrt{I} = \left\{f \mid f^m \in I \hbox{ for some integer }m\geq 1\right\}.$ We say that $I$ is a {\it radical ideal} if $I = \sqrt{I}$.

Let us recall some basics on the correspondence between ideals and varieties. Given an affine variety $V \subseteq k^m$ we can define the ideal of all polynomials vanishing on $V$, i.e.
$$\mathcal I(V) = \left\{ f \in k[X_1, \ldots, X_m] \mid  f(x) = 0 \hbox{ for all } x\in V \right\}.$$
Conversely, given an ideal $I\subseteq k[X_1, \ldots, X_m]$ we can define the affine variety
$$\mathcal V(I) = \left\{ x\in k^m \mid f(x) = 0 \hbox{ for all } f \in I\right\}.$$

Hilbert's Nullstellensatz (see, e.g., \cite[Theorem 6]{Cox-07}) states that if $k$ is algebraically closed and $I$ is an ideal in $k[X_1, \ldots, X_m]$, then $\mathcal I (\mathcal V(I)) = \sqrt{I}$. In particular, this implies that if we restrict to radical ideals, then the above maps are inverses of each other and we have a one-to-one correspondence between affine varieties and radical ideals. 

Let $K$ be the algebraic closure of $k$ and let $I$ be a zero-dimensional ideal, we define the quotient ring $R = K[X_1, \ldots, X_m]/I$. Then \cite[Theorem 2.10]{Cox-05} shows that the dimension of $R$ as a $K$-vector space gives a bound on the number of points in $\mathcal V(I)$. 
That is, 
$$\dim_{K}(R)\geq \# \mathcal V(I) \hbox{, with equality if and only if }I \hbox{ is a radical ideal;}$$
where $\# A$ denotes the cardinality of the set $A$.

Notice that $I$ is a zero-dimensional ideal if and only if $\mathcal V(I)$ is a finite set, that is $\mathcal V(I) = \left\{ P_1, \ldots, P_n\right\}$. The key idea to prove this result is to show that the evaluation map $\varphi$ defined as
\begin{equation} \label{eq:epimorphism}\begin{array}{cccc} 
\varphi: & K[X_1,\ldots,X_m] & \longrightarrow & K^n \\
& f & \mapsto & \left(f(P_1), \ldots, f(P_n)\right) 
\end{array} \end{equation}
is an epimorphism of $K$-vector spaces and ${\rm Ker}(\varphi) = \mathcal I(\mathcal V (I)) = \sqrt{I}$.

Although we have introduced all the results for an arbitrary field, from now on we will work with the finite field with $q$ elements, denoted as $\mathbb F_q$.

Let $I\subseteq \mathbb F_q[X_1, \ldots, X_n]$ be an ideal, we define the ideal $I_q$ related to $I$ as
$$I_q = I + \left\langle X_1^q-X_1, \ldots, X_m^q-X_m\right\rangle \subseteq \mathbb F_q[X_1, \ldots, X_m].$$
It is easy to check that $I_q$ is radical as consequence of Seidenberg's Lemma (because $I_q$ contains a univariate square free polynomial in each of the $m$-variables). Moreover, the points of the affine variety defined by $I_q$ (over the algebraic closure of $\mathbb F_q$) are the $\mathbb F_q$-rational points of the affine variety defined by $I$. That is, 
$$\mathcal V_{\overline{\mathbb F_q}}(I_q) = \mathcal V_{\mathbb F_q}(I_q) = \mathcal V_{\mathbb F_q}(I) = \left\{ P_1, \ldots, P_n\right\}.$$
where $\overline{\mathbb F_q}$ denotes the algebraic closure of $\mathbb F_q$.

Now we consider the quotient ring 
$R_I = \mathbb F_q[X_1, \ldots, X_m]/I_q$ and denote $\mathcal P = V_{\mathbb F_q}(I)= \left\{ P_1, \ldots, P_n\right\}$. By (\ref{eq:epimorphism}), the following evaluation map at the points of $\mathcal P$ is an isomorphism of $\mathbb F_q$-vector spaces:
\begin{equation} \label{eq:isomorphism} \begin{array}{cccc} \mathrm{ev_{\mathcal P}:} & R_I & \longrightarrow & \mathbb F_q^n\\
& f+I_q & \longmapsto & (f(P_1), \ldots, f(P_n)).\end{array} \end{equation}

\begin{definition}
Let $I_q$ and $R_q$ be defined as before and let $L$ be an $\mathbb F_q$-vector subspace of $R_q$ we define the affine variety code $C(I,L)$ as the image of $L$ under the evaluation map $\mathrm{ev}_{\mathcal P}$. That is:
$$C(I,L) = \mathrm{ev}_{\mathcal P}(L) = \left\{ \mathrm{ev}_{\mathcal P}(f+I_q) \mid f + I_q \in L\right\}.$$
\end{definition}

It is clear that $C(I,L)$ has $G=\left(f_i(P_j) \mid i=1, \ldots, k \hbox{ , } j=1, \ldots, n\right)$ as generator matrix where $\{f_1, \ldots, f_k\}$ form a basis of $L$.

\begin{example}
Let $I=\left\langle X^{q-1}-1\right\rangle \subseteq \mathbb F_q[X]$. Then, $I_q=I$ and $\mathcal V_{\mathbb F_q}(I) = \mathbb F_q^*$. Consider $L = \left\langle 1, X, \ldots, X^{k-1}\right\rangle$, then $C(I,L)$ is the \emph{Reed-Solomon code} of dimension $k$ over $\mathbb F_q$, denoted as $\mathrm{RS}_q(k)$. Moreover, if we set $I=(0)$, then $\mathcal V_{\mathbb F_q}(I) = \mathbb F_q$ and $C(I,L)$ is the \emph{extended Reed-Solomon code} of dimension $k$.
\end{example}

\begin{example}
Let $I=(0)\subseteq \mathbb F_q[X_1, \ldots, X_m]$. Then $I_q  = \left\langle X_1^q-X_1, \ldots, X_m^q-X_m\right\rangle$ and  $\mathcal V_{\mathbb F_q}(I) = \mathbb F_q^m$. If we take 
$L=\left\{ f \in \mathbb F_q[X_1, \ldots, X_m] \mid \mathrm{deg}(f) < s\right\}$, then $C(I,L)$ is the \emph{$q$-ary Reed-Muller code of degree $s$ in $m$ variables}, denoted as $\mathrm{RM}_q(s,m)$.
\end{example}

The reader may have already realized that some of the well-known classes of evaluation codes can be viewed as affine variety codes. Moreover in \cite[Proposition 1.4]{F98} it is  proved that every $\mathbb F_q$-linear code $\mathcal C$ may be represented as an affine variety code over $\mathbb F_{q^s}$ where we have to choose $s$ so that $q^s$ is greater than the length of $\mathcal C$. 

Let $\mathcal C=C(I,L)$ be an affine variety code. Then, it is clear that the length of $\mathcal C$ is the cardinality of $\mathcal V_{\mathbb F_q}(I) = \mathcal P = \{P_1, \ldots, P_n\}$ and the dimension of $\mathcal C$ is the dimension of the subspace $L$ - since the evaluation map $\mathrm{ev}_{\mathcal P}$ is an isomorphism. In the rest of the section we will study the minimum distance of affine variety codes $\mathcal C=C(I,L)$ in the particular case that $I=(0)$.

Let $A \subseteq \mathbb N^m$ be a non-empty (finite) subset of $\mathbb N^m$. We denote by $\mathbb F_q[A]\subseteq \mathbb F_q[X_1, \ldots, X_m]$ the $\mathbb F_q$-vector space with basis:
$$\left\{ X_1^{i_1} \cdots X_m^{i_m} \mid (i_1, \ldots, i_m)\in A\right\}.$$
We will denote by $\mathcal C_A$ the affine variety code $C(I,L)$ with $I=(0)$ and $L=\mathbb F_q[A]$, in other words $\mathcal C_A$ consists of the evaluation of polynomials $f\in \mathbb F_q[A]$ in the $q^m$ points of $\mathbb F_q^m$. 
\begin{remark}
Let $A\subseteq \mathbb N^m$ and consider the code $\mathcal C_A$ as the affine variety code $C(I,L)$ with $I=(0)$ and $L=\mathbb F_q[A]$. Then the length of $\mathcal C_A$ is $q^m$ and its dimension coincides with the cardinality of the set $A$.
\end{remark}

For $a, b \in \R$ and $a \leq b$, we denote by $\lii a, b \rii$ the integer interval $[a,b] \cap \Z$.

\begin{remark} Given $A\subseteq \mathbb N^m$ and using the identity $z^q=z$ for every $z\in \mathbb F_q$, one can find a unique set $B \subseteq \lii0,q-1\rii^{m}$ such that $\mathbb F_q[B] + I_q = \mathbb F_q[A] + I_q$ and, thus, $A$ and $B$ define the same code $\mC_A = \mC_B$ (see Figure \ref{fig:Aq}).
This set will be denoted by $B=A_q$. Throughout the article we use both sets indistinctly.
\end{remark}

\begin{figure}[h!]
\centering
\includegraphics[scale=.8]{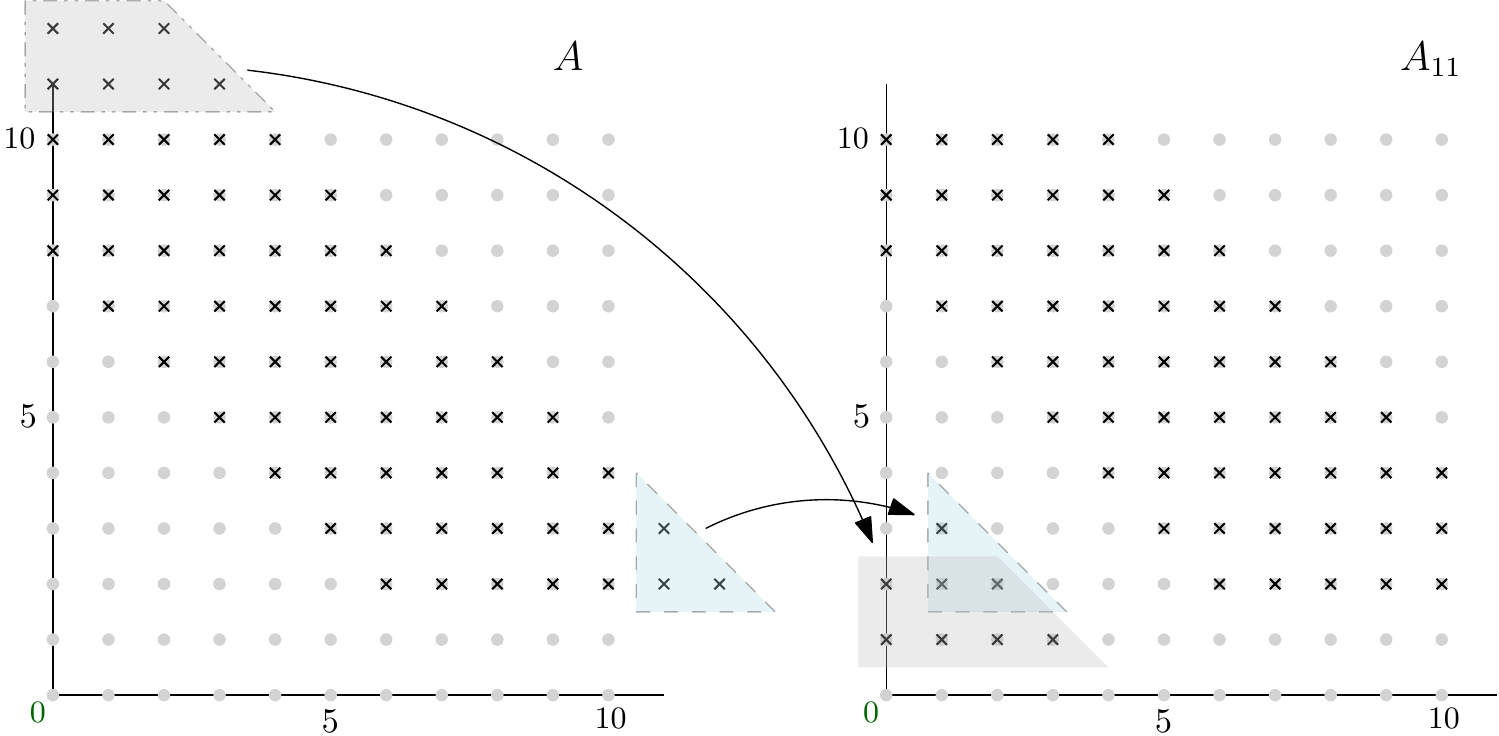}
\caption{The sets $A \subseteq \N^2$ and $A_{11} \subseteq \lii0,10\rii^2$ define the same code over $\F_{11}$.}
\label{fig:Aq}
\end{figure}
 
Let $f$ be a polynomial in $\mathbb F_q[X_1, \ldots, X_m]$, we define the ideal 
$$I_{q,f} = \left\langle X_1^q-X_1, \ldots, X_m^q-X_m , f\right\rangle$$ and the quotient ring $R_f = \mathbb F_q[X_1, \ldots, X_m] / I_{q,f}$.

\begin{proposition}
\label{preFB-1}
Let $f$ be a polynomial in $\mathbb F_q[X_1, \ldots, X_m]$, then the dimension of the $\mathbb F_q$-vector space $R_f$ is the number of roots of $f$ in $\mathbb F_q^m$. That is,
$\dim_{\mathbb F_q}(R_f) = \# \mathcal V_{\mathbb F_q}(f)$.
\end{proposition}

\begin{proof}
Applying (\ref{eq:isomorphism}) with $I = (f)$, then ${\rm \dim}_{\mathbb F_q}(R_f) = \# \mathcal V(I_{q,f}) = \# \mathcal V_{\mathbb F_q}(f)$. \end{proof}

The following well-known result (see, e.g., \cite{footprint}) gives a bound for the minimum distance of the particular case of affine variety codes of type $\mathcal C_A$. We include a short proof of this result for the sake of clarity.

\begin{theorem}[Footprint bound]
\label{FB}
Let $A\subseteq \lii0,q-1\rii^m$. Then, the minimum distance of $\mathcal C_A$  satisfies that
$$ d(\mathcal C_A) \geq \min_{(i_1, \ldots, i_m) \in A} \left\{ (q-i_1) \cdots (q-i_m) \right\}.$$
\end{theorem}

\begin{proof}
Since the codewords of $\mathcal C_A$ consists of the evaluation of polynomials $f\in \mathbb F_q[A]$ at the $n = q^m$ points of $\mathbb F_q^m$ and using the definition of minimum distance we have that
$$
d(\mathcal C_A)   =   n - \max_{f\in \mathbb F_q[A]} \# \left\{ \mathbb F_q\hbox{-roots of } f \right\} 
 = n -  \max_{f\in \mathbb F_q[A]}\# \mathcal V(I_{q,f}).$$
 Now using Proposition \ref{preFB-1} and standard Gr\"obner basis arguments  if we take $\succ$ any monomial order we have that 
\begin{eqnarray*}
d(\mathcal C_A) &=&  n - \max_{f\in \mathbb F_q[A]} \left\{  \dim_{\mathbb F_q}(R_f) \right\}
= n - \max_{f\in \mathbb F_q[A]} \left\{  \dim_{\mathbb F_q} \left(\mathbb F_q[X_1, \ldots, X_m] / \mathrm{in}(I_{q,f}) \right)\right\}\\
& \geq & n - \max_{f\in \mathbb F_q[A]} \left\{  \dim_{\mathbb F_q} \left(\mathbb F_q[X_1, \ldots, X_m] / \left\langle X_1^q, \ldots, X_n^q, \mathrm{in}(f)\right\rangle\right)\right\}\\
& = & n - \max_{(i_1, \ldots, i_m)\in A} \left\{  \dim_{\mathbb F_q} \left(\mathbb F_q[X_1, \ldots, X_m]/\left\langle X_1^q, \ldots, X_n^q,  X_1^{i_1}\cdots X_m^{i_m}\right\rangle\right)\right\}\\
& = & \min_{(i_1, \ldots, i_m) \in A} \left\{ (q-i_1) \cdots (q-i_m) \right\}.
\end{eqnarray*}
\end{proof}

\begin{definition}
Let $A\subseteq \lii0,q-1\rii^m$. We define the \emph{footprint-bound} of the affine code $\mathcal C_A$ as the integer
$$\mathrm{FB}(\mathcal C_A) = \min_{(i_1, \ldots, i_m) \in A} \left\{ (q-i_1) \cdots (q-i_m) \right\}.$$
By Theorem \ref{FB}, we have that the minimum distance of the code $\mathcal C_A$ satisfies that $$d(\mathcal C_A) \geq \mathrm{FB}(\mathcal C_A).$$
\end{definition}

In the following lines we study some well-known families of affine codes, namely (weighted) Reed Muller and hyperbolic codes (see, e.g., \cite{Sorensen:1992}, \cite{FR95},  \cite{GH01}). 
All these codes are konwn to satisfy that their minimum distance coincides with the value of the footprint-bound. One could provide an alternative proof of this fact by a direct application of Lemma \ref{FB-1} in the Appendix.

\begin{definition} \textbf{(Reed-Muller codes)}
Let $s\in \mathbb N$ and
$$A = \left\{ (i_1, \ldots, i_m)\in \lii0,q-1\rii^m \mid i_1 + \ldots + i_m \leq s\right\}.$$ Then, $\mathcal C_A$ is the called the $q$-ary \emph{Reed-Muller} code of degree $s$ in $m$ variables and we denote it by $\mathrm{RM}_q(s,m)$.
\end{definition}

The following result is known and the proof can be found in \cite[Theorem 2] {G08}. 
\begin{proposition}\label{pr:distRM} Given $s\in \mathbb N, s \leq (q-1)m$.  If we write  $s=a(q-1)+b$ with $0\leq b \leq q-1$, then the minimum distance of the Reed-Muller code $\mathcal C=\mathrm{RM}_q(s,m)$ is 
$$d(\mathcal C) = (q-b) q^{m-1-a}.$$
\end{proposition}


\begin{definition} \textbf{(Weighted Reed-Muller codes)}
Consider $s, s_1, \ldots, s_m > 0$ and let $A = \left\{ (i_1, \ldots, i_m)\in \lii 0,q-1\rii^m \mid s_1i_1 + \ldots + s_mi_m \leq s\right\}$. Then, $\mathcal C_A$ is called the $q$-ary \emph{weighted Reed-Muller} code of degree $s$ in $m$ variables with $\mathcal S = \left( s_1, \ldots, s_m \right)$ and we denote it by $\mathrm{WRM}_q(s,m, \mathcal S)$. If $s_1 = \ldots = s_m = 1$, then $\mathrm{WRM}_q(s,m, \mathcal S)$ is the corresponding $q$-ary Reed-Muller code $\mathrm{RM}_q(\lfloor s \rfloor,m)$.
\end{definition}

\begin{definition} \textbf{(Hyperbolic codes)}
Let $d\in \mathbb N$ and $$A=\left\{ (i_1, \ldots, i_m)\in \lii0,q-1\rii^m \mid (q-i_1) \cdots (q-i_m) \geq d\right\}.$$ Then, $\mathcal C_A$ is called the $q$-ary \emph{hyperbolic} code of order $d$ and we denote it by $\mathrm{Hyp}_q(d,m)$.
\end{definition}

\begin{figure}[h!]
    \centering
    \begin{subfigure}[b]{0.4\textwidth}
        \includegraphics[scale=0.8]{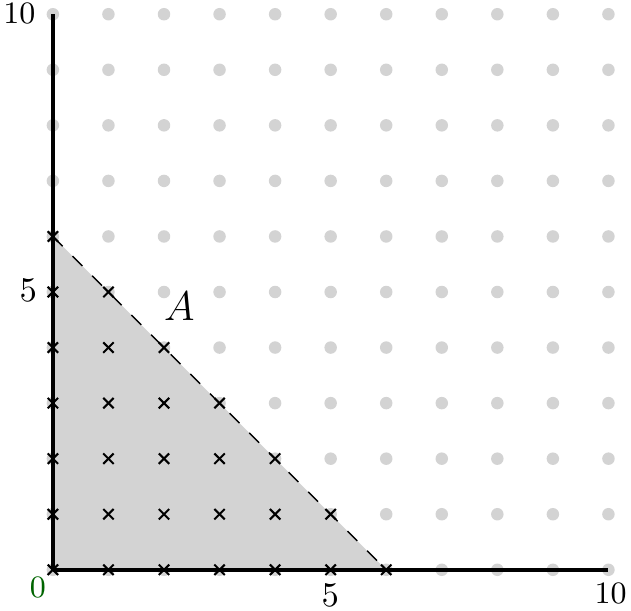}
        \caption{Example of a Reed-Muller code.}
    \end{subfigure}
    ~ 
    \begin{subfigure}[b]{0.55\textwidth}
        \includegraphics[scale=0.8]{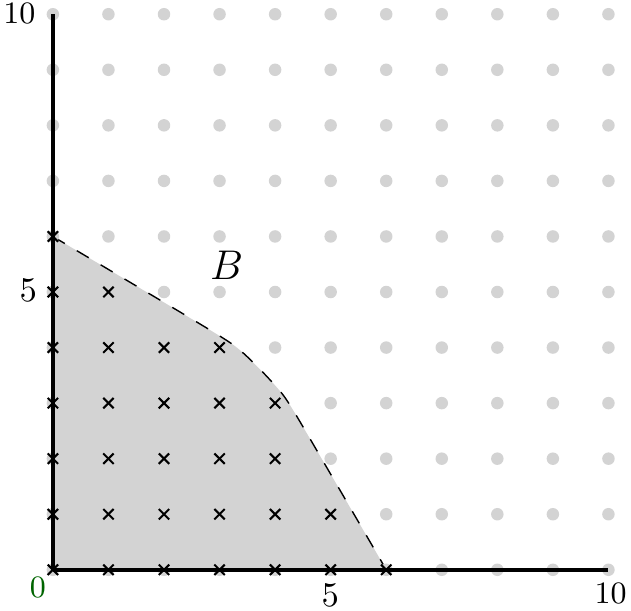}
        \caption{Example of a hyperbolic code.}
    \end{subfigure}\\
    
        \begin{subfigure}[b]{0.55\textwidth}
        \includegraphics[scale=0.8]{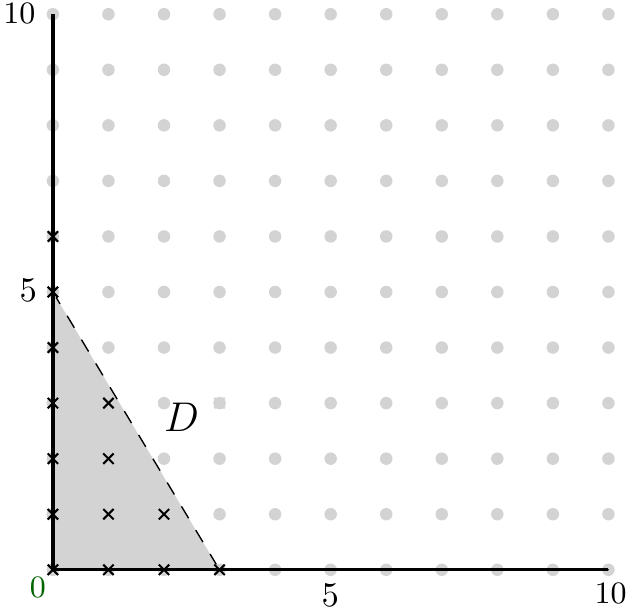}
        \caption{Example of a weighted Reed-Muller code.}
    \end{subfigure}
    \caption{Examples of Reed-Muller, hyperbolic and weighted Reed-Muller  codes.}
    \label{fig:examplesofcodes}
\end{figure}

\begin{example} Consider the following codes over $\mathbb F_{11}$
(see Figure \ref{fig:examplesofcodes}):
\begin{itemize}
\item the set
$A=\left\{ (i,j) \in
\lii0,10\rii^2 \mid i+j\leq 6\right\}$, corresponds to the Reed-Muller
code $\mC_A = \mathrm{RM}_{11}(6,2)$ with parameters $[11^2, 28,
55]_{11}$,  
\item the set $B=\left\{ (i,j) \in \lii0,10\rii^2 \mid
(11-i)(11-j)\leq 55\right\}$, corresponds to the hyperbollic code
$\mC_B = \mathrm{Hyp}_{11}(55,2)$ with parameters $[11^2, 30,
55]_{11}$,
\item and the set $D=\left\{ (i,j) \in \lii0,10\rii^2 \mid 5i +
3j \leq 15\right\}$, corresponds to the  weighted Reed-Muller code
$\mC_D = \mathrm{WRM}_{11}(15,2,\{5,3\})$ with parameters $[11^2, 13,
66]_{11}$.
\end{itemize}
 \end{example}

The hyperbolic code $\mathrm{Hyp}_q(d,m)$ has been designed to be the code with the largest possible dimension among those affine codes $\mC_A$ such that $\mathrm{FB}(\mC_A) \geq d$. 
In the following result, we indicate in the case of two variables, when the hyperbolic code of order $d$ has greater dimension with respect to a Reed-Muller code with the same minimum distance $d$.

\begin{proposition}
Consider $\mathcal D=\mathrm{RM}_q(t,2)$ and $\mathcal E=\mathrm{Hyp}_q(d,2)$. If $d(\mathcal D) = d(\mathcal E)$, then $k(\mathcal D) \leq k(\mathcal E)$. Moroever, $k(\mathcal D) < k(\mathcal E)$ if and only if 
$$\frac{t+5}{2}\leq q \leq \frac{(t+1)^2}{4}.$$
\end{proposition}

\begin{proof}
Since $d(\mathcal D) = d(\mathcal E)$, we have that $ {\rm FB}(\mathcal E)= d(\mathcal E) = d(\mathcal D) =  {\rm FB}(\mathcal D)$. 
Set $M=(m_{i,j})_{0 \leq i, j \leq q-1}$ the matrix with $m_{i,j} = (q-i) (q-j)$.
 We have that $\mathcal D = \mathcal C_A$ and $\mathcal E =  \mathcal C_B$ with $A = \{(i,j) \in \lii 0, q-1 \rii \, \vert \, i + j \leq t\}$ and
$B = \{(i,j) \in \lii 0, q-1 \rii \, \vert \, m_{i,j} \geq d\}$.  Moreover, \[ {\rm min}\{m_{i,j} \, \vert \, (i,j) \in A\} = d(\mathcal D) = d(\mathcal E) = {\rm min}\{m_{i,j} \, \vert \, (i,j) \in B\}. \]
Hence, $A \subseteq B$ and $k(\mathcal D) \leq k(\mathcal E)$; indeed, this proves that hyperbolic codes have the maximum dimension among all the codes with the same footprint-bound value.

By Proposition \ref{pr:distRM} we also have that \[ d(\mathcal D) =  \left\{ \begin{array}{lll} m_{0,t} & \text{ if \ } t \leq q-1 \text{,\ and} \\  
 m_{q-1,t-q+1} & \text{  if \ } q \leq t \leq 2q-2; \end{array} \right.\]   and it is easy to verify that
 \[ {\rm max}\{m_{i,j} \, \vert \, (i,j) \notin A\} = \left\{ \begin{array}{lll} m_{\frac{t+1}{2},\frac{t+1}{2}} & \text{if \ } t \text{ \ is\ odd,\ and }\\ 
 m_{\frac{t}{2},\frac{t+2}{2}}& \text{if \ } t \text{ \  is \ even }\end{array} \right.\]

We separate the proof depending on the value and the parity of $t$.
\begin{enumerate} 
\item $t \leq q-1$ and

\begin{enumerate}

\item $t$ is odd. Then $k(\mathcal D) < k(\mathcal E)$ if and only if $(\frac{t+1}{2},\frac{t+1}{2}) \in B$ or, equivalently, if 
$m_{t,0} \leq  m_{\frac{t+1}{2},\frac{t+1}{2}}$. Moreover, this happens if and only if
$q \leq \left( \frac{t+1}{2}\right)^2.$ 

\item  $t$ is even. Then $k(\mathcal D) < k(\mathcal E)$ if and only if  $(\frac{t}{2},\frac{t+2}{2}) \in B$ or, equivalently, if 
$m_{t,0} \leq m_{\frac{t}{2},\frac{t+2}{2}}$.
Moreover, this happens if and only if $q \leq \frac{t(t+2)}{4}.$ Since $t$ is even, this is equivalent to $q \leq \left(\frac{t+1}{2}\right)^2.$
\end{enumerate}
\item $t\geq q$ and
\begin{enumerate}
\item $t$ is odd. Then $k(\mathcal D) < k(\mathcal E)$ if and only if $(\frac{t+1}{2},\frac{t+1}{2}) \in B$ or, equivalently, if 
$m_{q-1,t-q+1} \leq  m_{\frac{t+1}{2},\frac{t+1}{2}}$. Moreover, this happens if and only if
$2q  - t - 1 \leq \left( q - \frac{t+1}{2}\right)^2.$ This defines a quadratic inequality $p(q) > 0$ involving in the variable $q$. Notice that
$p(q) \geq 0$ if and only if $q \leq (t+1)/2$ or $q \geq (t+5)/2$. The first option is not viable since $t \geq 2q-1$. We conclude, thus, that  $q  \geq  (t+5)/2.$

\item  $t$ is even. Then $k(\mathcal D) < k(\mathcal E)$ if and only if  $(\frac{t}{2},\frac{t+2}{2}) \in B$ or, equivalently, if 
$m_{q-1,t-q+1} \leq m_{\frac{t}{2},\frac{t+2}{2}}$.
 Moreover, this happens if and only if
$2q  - t - 1 \leq  \left( q - \frac{t}{2}\right)\left( q - \frac{t+2}{2}\right).$ Proceeding as in the previous case we get that this is equivalent to $q  \geq  \frac{t+3 + \sqrt{5}}{2}$ and since $t$ is even, this is the same as  $q \geq  \frac{t+5}{2}.$

\end{enumerate}
\end{enumerate}
\end{proof}

\section{Schur product of codes}
\label{Section3}

The notion of Schur product of codes was first introduced in coding theory for decoding \cite{pellikaan:1992} \cite{pellikaan:1996}.
But this operation turns out to have many other applications in cryptanalysis, multiparty computation, secret sharing or construction of lattices. Many of these applications are summarized in \cite[\S 4]{Hugues:2015}.

\begin{definition}
The Schur product is the componentwise product on $\F_q^n$. That is, given two elements $\mathbf a, \mathbf b \in \F_q^n$:
$$\mathbf a * \mathbf b \eqdef \left(a_1b_1, \ldots, a_nb_n \right)$$
For two codes $\mathcal{C}_1,\mathcal{C}_2\subseteq \F_q^n$, their Schur product is the code $\mathcal{C}_1*\mathcal{C}_2$ defined as
$$\mathcal{C}_1*\mathcal{C}_2 \eqdef \mathrm{Span}_{\F_q}\left\{\mathbf{c}_1 * \mathbf{c}_2 \mid \mathbf{c}_1 \in \mathcal{C}_1 \hbox{ and } \mathbf{c}_2 \in \mathcal{C}_2 \right\}$$

For $\mathcal{C}_1=\mathcal{C}_2=\mathcal{C}$, then $\mathcal{C}*\mathcal{C}$ is denoted as $\mathcal{C}^{\, (2)}$.
\end{definition}

\subsection{Product of codes and the Minkowski sum}
Given two sets $A, B \subseteq \N^m$, we denote by $A + B$ its Minkowski sum, that is, $A + B  =  \{a + b \, \vert \, a \in A, b \in B\}$ .
The following property is easy to check.

\begin{proposition}\label{pr:minkcuadrado} $\mC_A^{\, (2)} = \mC_{A + A}$.
\end{proposition} 
\begin{proof}
Let $\mathbf c\in \mC_A^{\, (2)}$, then $\mathbf c = \mathbf c_1 * \mathbf c_2$ with $\mathbf c_1, \mathbf c_2\in \mC_A$. Or equivalently, 
$$\mathbf c = \mathrm{ev}_{\mathcal P} (f) * \mathrm{ev}_{\mathcal P} (g) = \mathrm{ev}_{\mathcal P} (fg)\hbox{ with }f,g\in \mathbb F_q[A].$$
It is easy to check that if $f,g\in \mathbb F_q[A]$, then $fg\in \mathbb F_q[A+A]$. Thus, $\mathbf c \in \mC_{A+A}$.

Conversely, take notice that $\mathbb F_q[A+A]$ is the $\mathbb F_q$-vector space with basis 
$$
\left\{\mathbf X^{\mathbf i} \eqdef X_1^{i_1} \cdots X_m^{i_m} \mid \mathbf i = (i_1, \ldots, i_m)\in A+A \right\} = 
\left\{ \mathbf X^{\mathbf a} \cdot \mathbf X^{\mathbf b}
\mid \mathbf a, \mathbf b \in A\right\}$$
Therefore, for any $\mathbf c\in \mC_{A+A}$, then $\mathbf c = \mathrm{ev}_{\mathcal P}(f)$ with $f\in \mathbb F_q[A+A]$, that is 
$$\mathbf c = \mathrm{ev}_{\mathcal P}(f) = \mathrm{ev}_{\mathcal P}\left( \sum_{i=0}^s \lambda_i \mathbf X^{\mathbf a_i} \mathbf X^{\mathbf b_i} \right) = 
\sum_{i=0}^s \lambda_i \mathrm{ev}_{\mathcal P}(\mathbf X^{\mathbf a_i}) * \mathrm{ev}_{\mathcal P}(\mathbf X^{\mathbf b_i}) \in \mC_{A}^{\, (2)}.$$
\end{proof}

It is important to highlight that even if $A \subset \lii0,q-1\rii^m$, it might happen that  $A + A \not\subset \lii0,q-1\rii^m$; however 
$A'  := (A+A)_q \subset \lii0,q-1\rii^m$ satisfies that $\mC_{A'} = \mC_{A + A}$ (see Figure \ref{fig:amasa}).

\begin{figure}[h!]
\centering
\includegraphics[scale=.6]{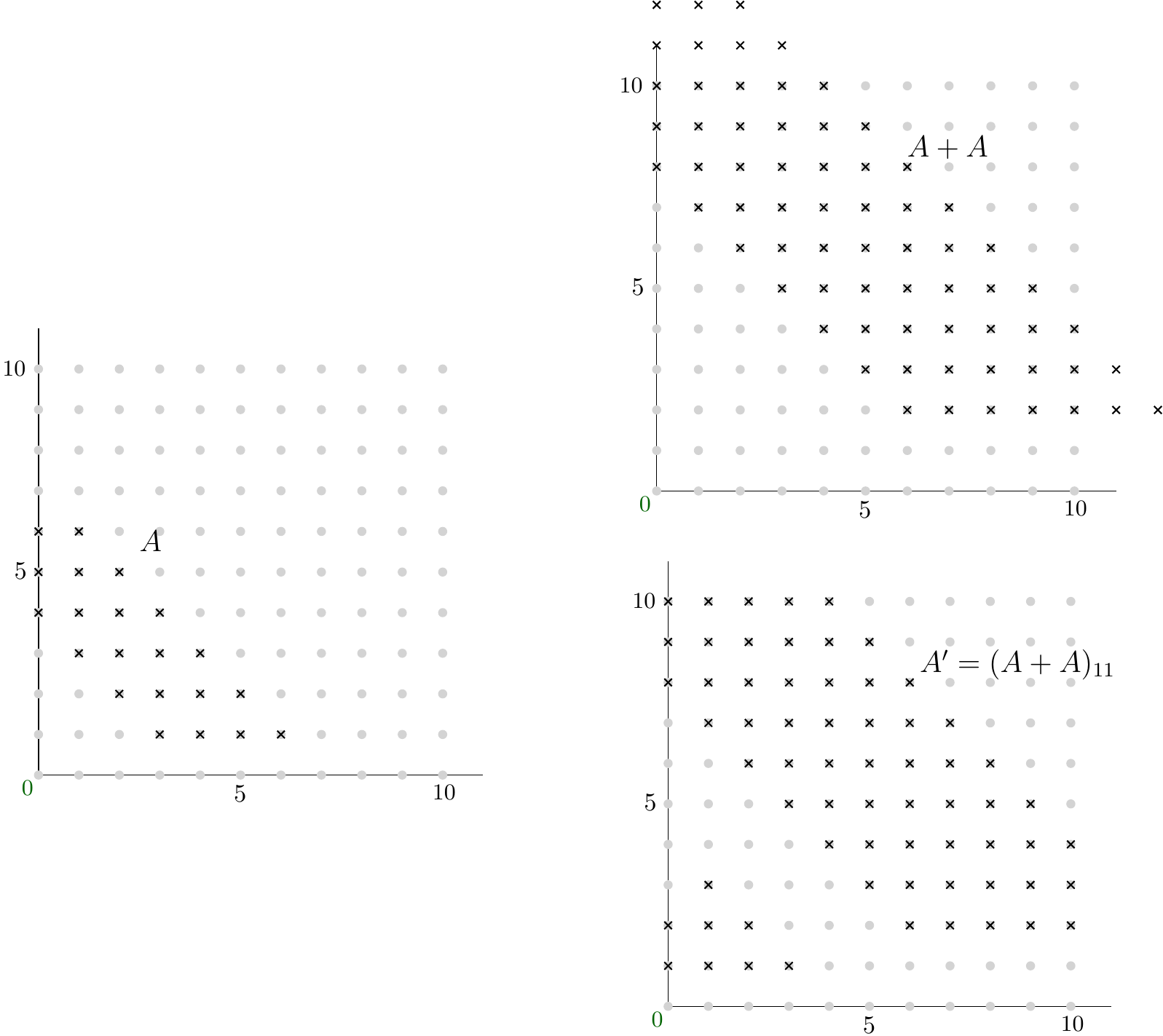}
\caption{By Proposition \ref{pr:minkcuadrado} we have that $\mC_A^{\, (2)} = \mC_{A + A} = \mC_{A'}$.}\label{fig:amasa}
\end{figure}

Proposition \ref{pr:minkcuadrado} suggests the following way of constructing affine codes whose square has a designed minimum distance: we consider 
$B \subset \lii0,q-1\rii^m$ such that $d(\mC_{B}) \geq d$ and, then, we choose $A$ such that $(A + A)_q \subset B$.  If $A$ is chosen in this way, then we will have that
$d(\mC_{A}^{\, (2)}) =  d(\mC_{A+A}) = d(\mC_{(A+A)_q}) \geq d(\mC_{B}) \geq d$.
The following 
lemma gives a necessary condition for such a set~$A$.

\begin{lemma} Let $A,B \subset \lii0,q-1\rii^m$ and for each $\epsilon = (\epsilon_1,\ldots,\epsilon_m) \in \{0,1\}^m$ we set  
\[ B_{\epsilon} := \{ \mathbf{b} + (q-1) \epsilon \, \vert \, \mathbf{b} = (b_1,\ldots,b_n) \in B {\text \  and \ } b_i > 0 {\text\ whenever \ } \epsilon_i = 1\}. \]  If $(A+A)_q \subseteq B$, then $2A = \{2 \mathbf{a} \, \vert \, \mathbf{a} \in A\}$ is a subset of $\cup_{\epsilon \in \{0,1\}^m} B_{\epsilon}.$ 
\end{lemma}
\begin{proof}
Assume that $(A+A)_q \subseteq B$. 

We observe that  whenever $\mathbf{a} = (a_1,\ldots,a_m) \in A,$ then $(2 \mathbf{a})_q \in (A+A)_q$, 
$$\hbox{where  }(2 \mathbf{a})_q  = (a_1',\ldots,a_m')\hbox{ with }
a_i' = \left\{ \begin{array}{ll} 2 a_i, & \hbox{ if } 2 a_i < q,\\ 
2 a_i - (q-1) & \hbox{ otherwise. }\end{array} \right.$$ 
Now, it suffices to take 
$$\epsilon = (\epsilon_1,\ldots,\epsilon_m)
\hbox{ with }\epsilon_i = \left\{ \begin{array}{ll} 0 & \hbox{ if }2a_i < q\\
1 & \hbox{ otherwise}
\end{array}\right.$$
to have that $2 \mathbf{a} \in B_{\epsilon}$. 
 \end{proof}

\begin{figure}[h!]
\centering
\includegraphics[scale=.6]{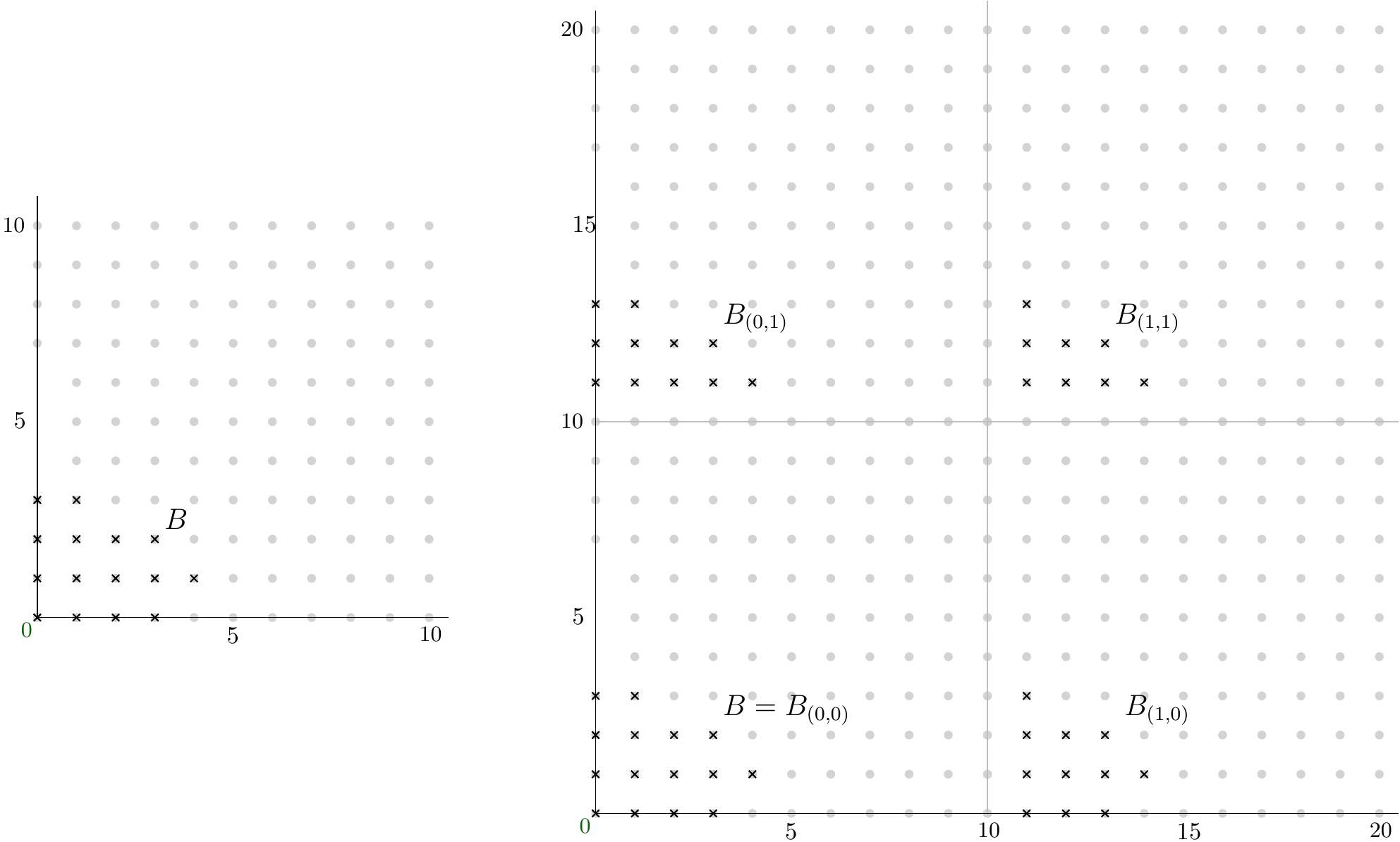}
\caption{Examples of sets $B_{\epsilon}$ for $B \subseteq \lii0,10\rii^2$.}\label{fig:bepsilon}
\end{figure}

The following proposition and the subsequent theorem are the key results to understand our strategy to give a code $\mC_A$ whose square has designed minimum distance. They are both based
on (simple) convexity arguments.  
Given a set $B\subseteq \lii0,q-1\rii^m$, suppose that we want to find a set $A\subseteq \lii0,q-1\rii^m$ such that $(A+A)_q\subseteq B$. If such condition happens then we will have that $(2A)_q\subseteq (A+A)_q \subseteq B$. However the following lemma allows us to construct a set $A$ with the property that by just checking that $(2A)_q\subseteq B$, it will imply that  $(A+A)_q\subseteq B$.

\begin{proposition}\label{lema:convexo}Let $D \subset \R^m$ be a convex set and consider $A := \{\mathbf{a} \in \Z^m \, \vert \, 2 \mathbf{a} \in D\}$. Then, $A + A \subseteq D$.
\end{proposition}
\begin{proof}It suffices to check that $\mathbf{a} + \mathbf{a'} \in D$ whenever $\mathbf{a}, \mathbf{a'}  \in A$. By definition of $A$ we have that $2\mathbf{a}, 2\mathbf{a'}  \in D$ and, since $D$ is convex, the midpoint of the segment joining $2\mathbf{a}$ and $2\mathbf{a'}$, which is $\mathbf{a} + \mathbf{a'}$, also belongs to $D$.
\end{proof}

\begin{theorem}\label{pr:convexgeneral}  Let $d \in \N$ and let $\mC_B$ be a linear code with $B  \subseteq \lii0,q-1\rii^m$ such that $d \leq d(\mC_B)$. 
Consider $C \subset \R^m$ a convex set such that
\[ C \cap \{ \mathbf{c} \in \R^m\, \vert \, 2 \mathbf{c} \in \lii0,2q-2\rii^m {\text \ and \ } [2 \mathbf{c}]_q \notin B\} = \emptyset. \]

Taking $A := C \cap \lii0,q-1\rii^m$ we have that $d(\mC_A^{\, (2)}) \geq d$.
\end{theorem}
\begin{proof} To prove the statement we will just verify that $(A + A)_q \subseteq B$ and, hence, $d(\mC_{A}^{\, (2)}) = d({\mC_{A+A}}) \geq d(\mC_B) \geq d.$ 
Let us take $\mathbf{a}, \mathbf{a'} \in A$, we have that $A \subset C$ and $C$ is a convex set, so $(\mathbf{a} + \mathbf{a'})/2 \in C$. Thus, 
$[\mathbf{a} + \mathbf{a'}]_q \in B$.
\end{proof}

This result suggests a technique to obtain, for a given $d \in \N$, a set $A$ such that $d(\mC_{A}^{\, (2)}) \geq d$, see Algorithm \ref{Algor-1}. Indeed, it suffices to consider a linear code $\mC_B$ such that
$d \leq d(\mC_B)$, choose a convex set  $C \subset \R^m$ satisfying the hypotheses of the previous result and then, take  $A := C \cap \lii0,q-1\rii^m$. If one wants to have a large value of $k(\mC_A)$ one has to choose $C$ strategically so that it has the maximum number of integer points.

\begin{algorithm}[H] \label{Algor-1}
\SetAlgoLined
\KwData{A set $B\subseteq \lii0,q-1\rii^m$ such that $d(\mC_B)\geq~d$}
\KwResult{An affine variety code $\mC_A$ such that $d(\mC_A^{\, (2)})\geq d$}
\begin{algorithmic}
\STATE Choose a convex set $C\subseteq \lii0,q-1\rii^m$ satisfying that:$$C \cap \{ \mathbf{c} \in \Q^m\, \vert \, 2 \mathbf{c} \in \lii0,2q-2\rii^m {\text \ and \ } [2 \mathbf{c}]_q \notin B\} = \emptyset.$$
 Take $A=C\cap \mathbb N^m$
\caption{Procedure to find a set $A\subseteq \lii0,q-1\rii^m$ with $d(\mC_A^{\, (2)})\geq~d$.}
\end{algorithmic}
\end{algorithm}

\medskip

In particular, if we apply the previous result to $\mC_B$ a hyperbolic code of order $d$, we get the following.

\begin{proposition}\label{pr:convexhyp} Let $C \subset \R^m$ be a convex set such that $C \cap D_{\epsilon} = \emptyset$ for  all 
$\epsilon \in \{0,1\}^m$, being 

$$
D_{\epsilon}  =   \{ (b_1,\ldots,b_m) \ \vert \ 2 b_i \in  \left\{ \begin{array}{cll} \left\lii 0,q-1\right]  & {\text if} & \epsilon_i = 0 \\  \left\lii q,2q-2\right\rii & {\text if} & \epsilon_i = 1 \end{array} \right.  \ {\rm and}\ \prod_{i = 1}^m   (q+\epsilon_i(q-1)-2b_i) < d\} 
$$
 
Then, taking $A := C \cap \lii 0,q-1\rii^m$ we have that $d(\mC_A^{\,(2)}) \geq d$.

\end{proposition}
\begin{proof} Take $B \subseteq \lii 0,q-1\rii^m$ such that
$\mC_B = \mathrm{Hyp}_q(d,m)$; then we have that $d(\mC_{B}) \geq d$. Taking into account the following equation
\[ \{ \mathbf{c} \in \R^m\, \vert \, 2 \mathbf{c} \in \lii 0,2q-2\rii^m {\text \ and \ } [2 \mathbf{c}]_q \notin B\} = \cup_{\epsilon \in \lii 0,1\rii^m} D_{\epsilon}.\]
and Theorem  \ref{pr:convexgeneral} the result holds.
\end{proof}

Let us illustrate this result with an example, see Figure \ref{fig:convexhyp} for a graphic representation.
\begin{example}\label{ex:convexhyp} Consider $q = 11,\, m = 2$ and $d = 6$. We are going to construct a code $\mC_A$ over $\F_{11}$ such that $d(\mC_A^{\,(2)}) \geq 6$, following Proposition \ref{pr:convexhyp}. Consider $\mC_B = \mathrm{Hyp}_{11}(6,2)$ and $D_{\epsilon}$ for all $\epsilon \in \{0,1\}^2$. We choose $C$ a convex set such that $C \cap D_{\epsilon} = \emptyset$ for all $\epsilon \in \{0,1\}^2$ and take $A = C \cap \lii 0,10\rii^2$ as in Figure \ref{fig:convexhyp}. Then, as we proved in Proposition \ref{pr:convexhyp}, we have that $(A + A)_{11} \subseteq B$ and, thus, $d(\mC_A^{\,(2)}) = d(\mC_{A+A}) = d(\mC_{(A+A)_{11}}) \geq d(\mC_B) = d(\mathrm{Hyp}_{11}(6,2))  = 6.$  
\begin{figure}[h!]
\centering
\includegraphics[scale=.8]{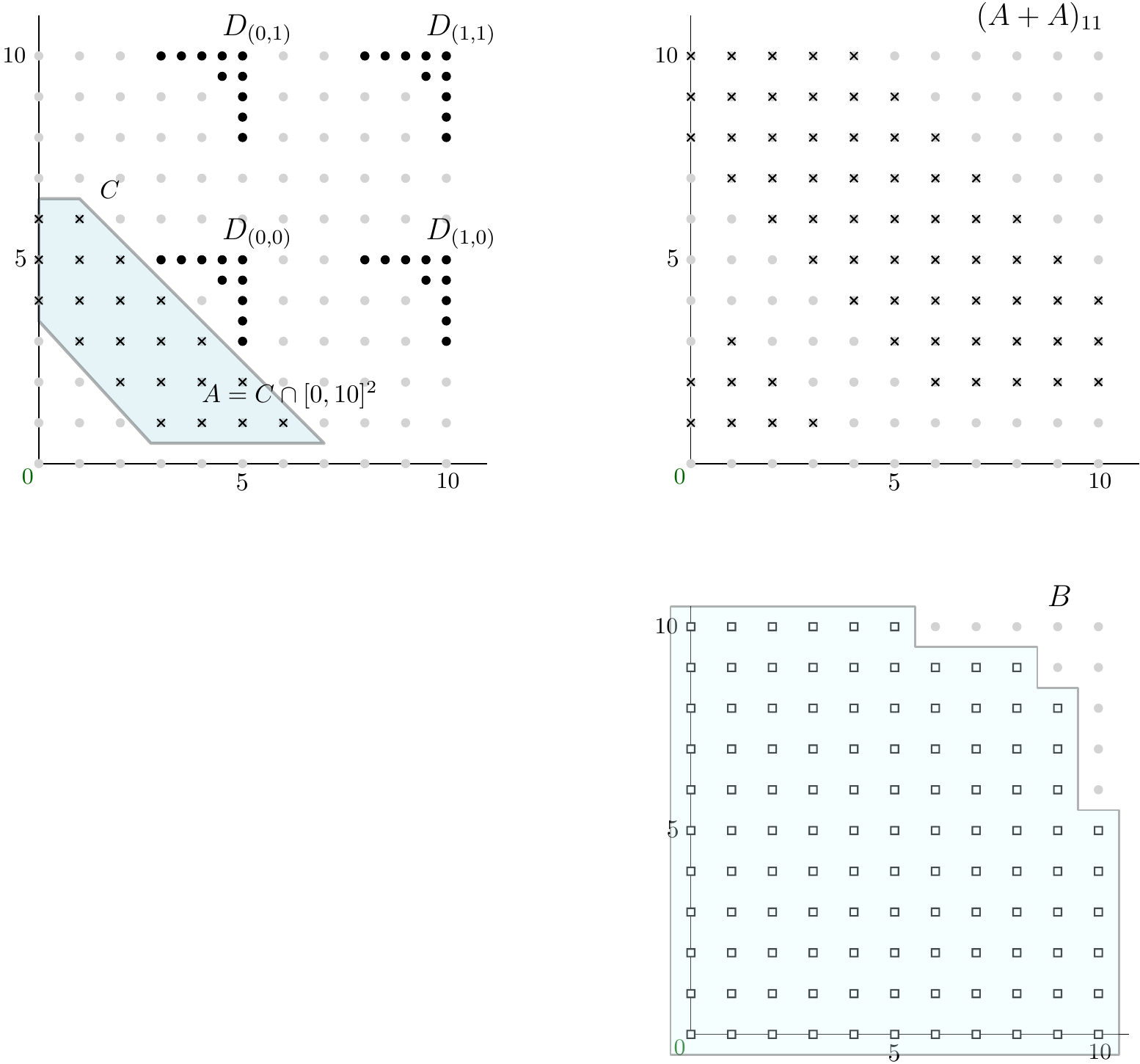}
\caption{Example of a code $\mC_A$ such that $(A+A)_{11} \subset B$ and, thus, $d(\mC_{A}^{\, (2)}) \geq d(\mC_B) = 6$  (see Example \ref{ex:convexhyp}).}  \label{fig:convexhyp}
\end{figure}
\end{example}

As one can see, following the construction of Proposition \ref{pr:convexhyp}, the number of integer points in the convex set $C$ tuns out to be the dimension of the code $A$ such 
that $d(\mC_{A}^{\, (2)}) \geq d$. So, in order to obtain a code $\mC_A$ with high dimension, one could look for convex sets with the most number of integer points possible. In the next section we are going to propose and compare several natural choices of the convex set $C$.

\section{Choosing a convex $C$ that gives affine codes with good parameters}
\label{Section4}

In the previous section we described in Algorithm \ref{Algor-1} a method that, given $d$, it returns a code $\mC_A$ such that $d (\mC_A^{\,(2)}) \geq d $. However, we would also like to find among all the codes $\mC_A$ that verify the previous property, the one that has the highest possible dimension. For this purpose, the convex set $C$ mentioned in Algorithm \ref{Algor-1} must have the maximum number of integer points.

Given a fixed value $d$, the hyperbolic code $\mC = \mathrm{Hyp}_q(d,m)$ of order $d$ is, by definition, the affine variety code with the highest  dimension among all the codes whose footprint-bound is $\geq d$. So it seems natural to run Algorithm \ref{Algor-1} being $B \subset \lii 0, q-1 \rii^m$ such that $\mC_B = \mathrm{Hyp}_q(d,m)$. Now, to choose $A$ such that $(A + A)_q \subset B$, it would be logical to expect that a certain code that behaves like a \emph{half hyperbolic}  code (see Definition \ref{def:halfhyp}) would be the best candidate  for our goal. Surprisingly, this is not always the case. As we will prove at the end of this section, when the value of $d$ is small enough there exist certain weighted Reed-Muller codes that outperform half hyperbolic codes.

\subsection{Half hyperbolic codes}
\label{Section4.1}

First let us  introduce half hyperbolic codes.

\begin{definition} \textbf{(Half hyperbolic codes)} \label{def:halfhyp}
Let $\mC_B = \mathrm{Hyp}_q(d,m)$ be an hyperbolic code with 
$B = \left\{ (i_1, \ldots, i_m)\in \lii 0,q-1\rii^m \mid (q-i_1) \cdots (q-i_m) \geq d\right\}$ and let 
$$A= \left\{ (i_1, \ldots, i_m)\in \blii 0,\frac{q-1}{2}\brii^m  \mid (q-2i_1) \cdots (q-2i_m)\geq d\right\}.$$
In other words, for $\mathbf{a} \in \blii 0,\frac{q-1}{2}\brii^m$, then $\mathbf{a} \in A$ if and only if $2\mathbf{a} \in B$.
Then, $\mC_A$ is the $q$-ary \emph{half hyperbolic} code of order $d$  and we denote it by 
$\mathrm{HalfHyp}_q(d,m)$.
\end{definition}

\begin{example}
\label{example:HHYP}
 \begin{figure}[h!]
    \centering
        \includegraphics[scale=0.8]{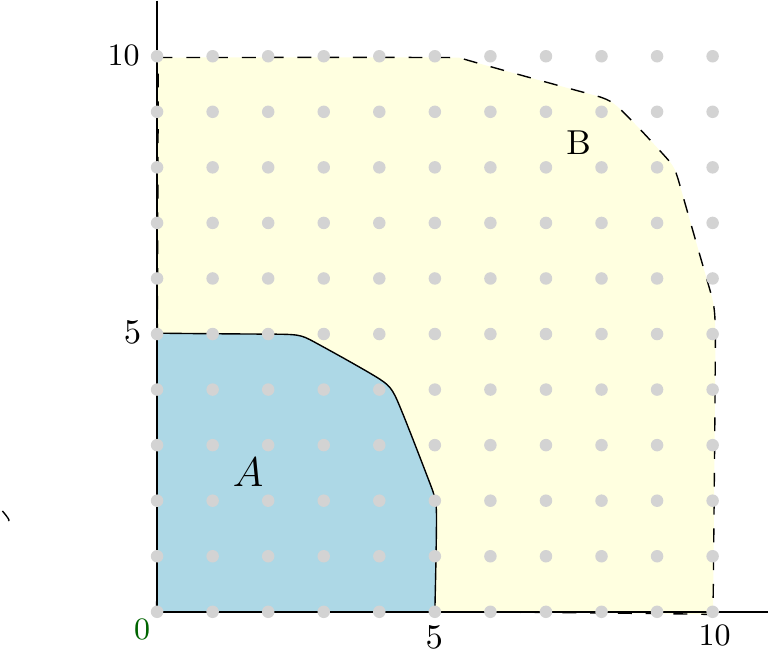}
    \caption{Figure illustrating Example \ref{example:HHYP}.}
    \label{fig:example-HHYP}
\end{figure}
Let $B=\left\{ (i,j) \in \lii 0,10\rii^2 \mid (11-i)(11-j) \geq 6\right\}$ then $\mC_B = \mathrm{Hyp}_{11}(6,2)$. The code $\mC_B$ has parameters $[11^2, 111,6]_{11}$.
Now we consider 
$$A=\left\{ (i,j) \in \lii 0,5\rii^2 \mid (11-2i)(11-2j) \geq 6\right\}$$
then $\mC_A$ is the half hyperbolic code of order $6$ and we denote it by $\mathrm{HalfHyp}_{11}(6,2)$. The code $\mC_A$ has parameters $[11^2, 25,49]_{11}$ and $d(\mC_A^{\, (2)})\geq 6$.
See Figure \ref{fig:example-HHYP} for a graphic representation of this example.
\end{example}

In the following result we use Proposition  \ref{pr:convexhyp} to prove  that the square of a half-hyperbolic code of order $d$ has minimum distance $\geq d$.

\begin{proposition}Let $d \in \Z^+$ such that $d < q^m$, then
$d(\mathrm{HalfHyp}_q(d,m)^{\, (2)}) \geq d$. 
\end{proposition}
\begin{proof}Taking $C = \{\mathbf{a} = (a_1,\ldots,a_m) \in \R^m \, \vert \ \ 0 \leq a_i \leq \frac{q-1}{2},\  \prod_{i = 1}^m (q-2a_i) \geq d\}$ we have that $C$ is a convex set. Moreover, by definition of this set $C \cap D_{\epsilon} = \emptyset$ for all $\epsilon \in \{0,1\}^m$ (where $D_{\epsilon}$ is defined as in Proposition \ref{pr:convexhyp}). Thus, taking $A = C \cap \lii 0,q-1\rii^m$, Proposition \ref{pr:convexhyp} guarantees that  
$d(\mC_A^{\,(2)}) \geq d$. To finish the proof it suffices to observe that $\mC_A$ coincides with $\mathrm{HalfHyp}_q(d,m)$.
\end{proof}

Providing a formula for the dimension of a half hyperbolic code is not an easy task. Nevertheless, we provide an expression for the dimension of a half hyperbolic code when $m = 2$:

\begin{lemma}Let $d \in \Z^+$ such that $d < q^2$, then
\[ k(\mathrm{HalfHyp}_q(d,2)) = \sum_{i = 0}^{\lfloor \frac{q^2-d}{2q} \rfloor} \left\lfloor \frac{d+(q+2)(2i-q)}{4i-2q} \right\rfloor \] 
\end{lemma}
\begin{proof}Since  $\mathrm{HalfHyp}_q(2,d) = \mC_{A}$ with $A = \{(i,j) \in \N^2  \, \vert \, 0  \leq i,j \leq (q-1)/2$ and  $(q-2i)(q-2j) \geq d\}$, then  
\[ k(\mathrm{HalfHyp}_q(2,d)) = |A|. \]
Moreover, setting $A_i := \{j \, \vert \, (i,j) \in A\}$ for all $i \in \lii 0, (q-1)/2 \rii$, one has that 
$ |A| =   \sum_{i = 0}^{(q-1)/2} |A_i|$ and   
\begin{eqnarray*}
A_i & = &  \left\{j \, \vert \, 
0 \leq j \leq (q-1)/2\ {\rm and}\ (q-2i)(q-2j) \geq d\right\} \\ 
& = & \left\{j \, \vert \,0 \leq j \leq (q-1)/2\ {\rm and}\  q - 2j \geq d / (q-2i )\right\} \\ 
& = &  \left\{j \, \vert \, 0 \leq j \leq \frac{d+q(2i-q)}{4i-2q}\right\}.    
\end{eqnarray*}
Hence $A_i = \emptyset$ whenever $i > (q^2 - d)/2q$; and $|A_i| = \left\lfloor \frac{d+q(2i-q)}{4i-2q} \right\rfloor + 1$ otherwise.
\end{proof}

When $d \geq q$, the sets $D_{\epsilon}$ in Proposition \ref{pr:convexhyp} seem to `divide' $\lii 0,q-1 \rii^m$ into $2^m$ regions. For this reason, we propose $\mC = {\rm HalfHyp}_q(d,m)$, the half hyperbolic code of order $d$, as a code with high dimension $k(\mC)$ and satisfying that $d(\mC^{(2)}) \geq d$ (see Figure \ref{fig:HHypwinsforbigd}). As we will see in the following subsection, when $d < q$ one can find better options in the family of weighted Reed-Muller codes.

 \begin{figure}[h!]
    \centering
        \includegraphics[scale=0.8]{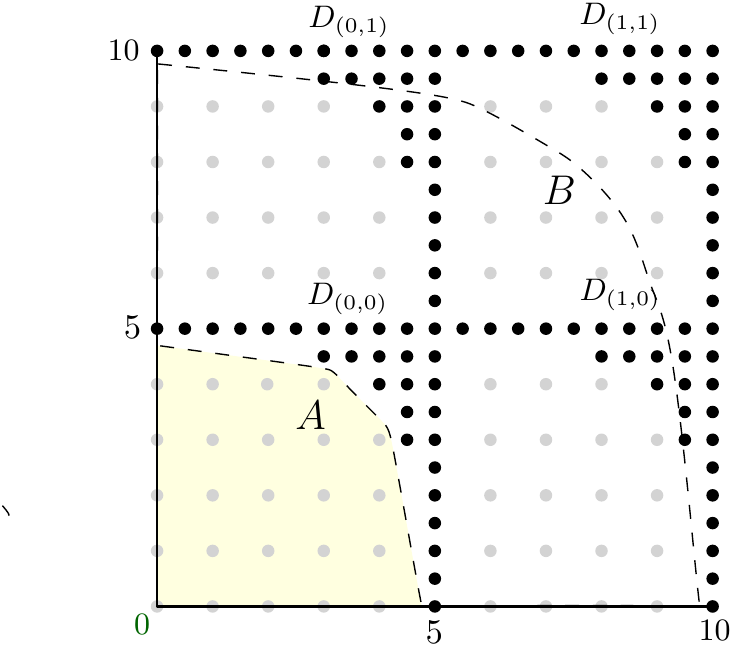}
    \caption{Example over $\mathbb F_{11}$ with $m = 2$ and $d = 12$. The sets $D_{\epsilon}$ described in Theorem \ref{pr:convexgeneral} seem to divide $\lii 0,10 \rii^2$ into
    4 regions. The code $\mC_A = {\rm HalfHyp}_{11}(11,2)$, which has parameters $[11^2, 24, 56]_{11}$  is a code with high dimension and $d(\mC_A^{(2)}) \geq 11$.}
    \label{fig:HHypwinsforbigd}
\end{figure}

\subsection{Weighted Reed-Muller codes}
\label{Section4.2}

It is not difficult to see that when $d \geq q$ and $\mC$ is a weighted Reed-Muller code with $d(\mC^{(2)}) \geq d$, then ${\rm FB}(\mC) \geq {\rm FB}({\rm HalfHyp}_q(m,d))$ and, hence, $k(\mC) \leq k({\rm HalfHyp}_q(m,d))$. As we will see at the end of this section, this is no longer true for all the values $d < q$, where some weighted Reed-Muller codes outperform half hyperbolic ones when $d$ is small enough (see Propositions \ref{prop:odd} and \ref{prop:even}). For simplicity this section concerns the case $m = 2$. Before proving Propositions \ref{prop:odd} and \ref{prop:even}, we characterize  which is the weighted Reed-Muller code with highest dimension among those verifying that the minimum distance of its square is at least $d$, provided $d < q$. It happens that the choice of this code depends on the parity of $d$ (see Theorem \ref{th:wrmdimpar} for $d$ odd, and Theorem \ref{th:wrmdpar} for $d$ even).

A first observation is that if $\mC$ is a weighted Reed-Muller code then $\mC^{\, (2)}$ is not necessarily a weighted Reed-Muller code as the following example shows.
\begin{example}
\label{Example:counterexample}
 \begin{figure}[h!]
    \centering
    \begin{subfigure}[b]{0.45\textwidth}
        \includegraphics[scale=0.8]{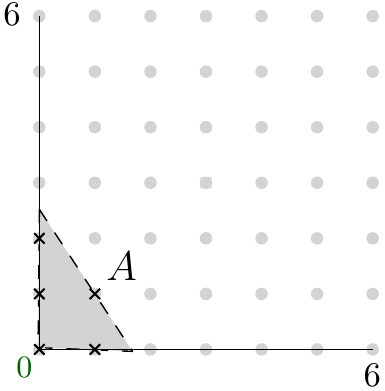}
        \caption{$A=\left\{ (i,j) \mid 3i+2j\leq 5\right\}$.}
    \end{subfigure}
    \begin{subfigure}[b]{0.45\textwidth}
        \includegraphics[scale=0.8]{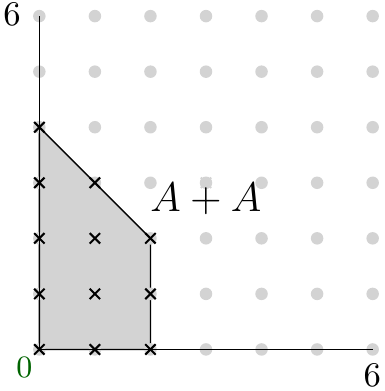}
        \caption{$A+A$ is marked with crosses.}
        \label{Fig:A+A}
    \end{subfigure}\\
    \caption{Figure illustrating Example \ref{Example:counterexample}.}
    \label{fig:lemma-WRM}
\end{figure}
Consider the weighted Reed-Muller code $\mC_A$ over $\mathbb F_7$ with 
$$A=\left\{ (i,j) \mid 3i+2j\leq 5\right\}.$$ Then the code $\mC_{A+A}$ is the affine variety code that consists of the evaluation of polynomials $f\in \mathbb F_q[A+A]$ in the points of $\mathbb F_7^2$, where $$A+A = \{(i,j) \, \vert \, 0 \leq i,j \leq 2\} \cup \{(0,3),(0,4),(1,3)\}$$ 
(see Figure \ref{Fig:A+A}). It is easy to check that there is no weighted Reed-Muller code $\mC_B$ such that $(2,2) \in B$, but $(0,5), (3,0) \notin B$. Thus, $\mC_{A+A}$ is not a weighted Reed-Muller code.
\end{example}

Despite the fact that the square of a weighted Reed-Muller is not necessarily a weighted Reed-Muller code, they verify the following property which will be important in the proofs
of the main results.
\begin{lemma}
\label{Lemma:WRM-sharp}
If $\mC$ is a weighted Reed-Muller code then $d(\mC^{\, (2)}) = {\rm FB}(\mC^{\, (2)})$.
\end{lemma}

\begin{proof}
Let $\mC_A$ be a weighted Reed-Muller code with $A \subset \lii 0,q-1 \rii^m$ and suppose that 
$$\mathrm{FB}(\mC^{\, (2)}) = \prod_{i = 1}^m (q-\alpha_i),$$
for some $\mathbf{a} = (\alpha_1, \ldots, \alpha_m)\in A+A$. Then $\mathbf{a} = \mathbf{b} + \mathbf{c}$ for some $\mathbf{b}, \mathbf{c} \in A$.
Taking the componentwise partial order $\leq$ in $\lii 0,q-1 \rii^m$ and $\mathbf{b'} \leq \mathbf{b}$ and $\mathbf{c'} \leq\mathbf{c}$,
 one has that   $\mathbf{b'}, \mathbf{c'} \in A$  because $\mC_A$ is a weighted Reed-Muller code. Then one easily gets that $\mathbf{a'} \in A + A$ for all $\mathbf{a'} \leq \mathbf{a}$ and  applying Lemma \ref{FB-1} we complete the proof.
\end{proof}

When $d \geq q$, it is easy to verify that the weighted Reed-Muller code with maximum dimension and designed minimum distance is a Reed-Muller code. Now we are going to characterize which are the weighted Reed-Muller codes with maximum dimension and designed minimum distance when $d < q$. We will have that it is also a Reed-Muller code when $d$ is odd, but instead, it is a weighted Reed-Muller one when $d$ is even. 

 \begin{theorem}\label{th:wrmdimpar} Let $\F_q$ be a finite field and $d \in \Z^+$ be an odd integer with $d < q$ and let $s := q-\frac{d+1}{2}$. If $\mC$ is a weighted Reed-Muller code over $\F_q$ with $d(\mC^{\, (2)}) \geq d$, then
$k(\mC) \leq k(\mathrm{RM}_q(2, s))$.
 \end{theorem}
\begin{proof} Let $\mC$ be a weighted Reed-Muller code over $\F_q$ with $d(\mC^{\, (2)}) \geq d$. We assume without loss of generality that $\mC = \mathrm{WRM}_q(\lambda,2,(w_1,1))$ for some $\lambda, w_1 > 0$. Taking 
$$A := \{(i,j) \in \lii 0,q-1 \rii \, \vert \, w_1 i + j \leq \lambda\}$$ we have that $\mC = \mC_A$. 

In this proof we denote $a:= (q-1)/2$ and $b := (q-d+1)/2$; and observe that $(2a,2b) \in \N^2$ and that $s = a + b - \frac{1}{2}.$

We divide the proof in two cases depending on the value of $\lambda$.

\textbf{Case I:} $\lambda \leq a + b + \frac{1}{2}$. If we consider  $B := \{(i,j) \in \N^2\, \vert \, i + j \leq s\}$, then $\mathrm{RM}_q(2,s) =  \mC_{B}$. To prove that $|A| = k(\mC) \leq k(\mathrm{RM}_q(2,s)) = |B|$ we are going to prove that either $A \subseteq B$, or  the symmetry through the point $(a,b)$:  
\[ \begin{array}{lccl} \varphi: & A - B &\ \longrightarrow\ & B - A \\ & (\alpha,\beta) & \mapsto & (2a-\alpha, 2b - \beta) \end{array}\]
is an injective map (see Figure \ref{fig:thwrmdimpar} for a graphic representation of this idea). 

 \begin{figure}[h!]
    \centering
        \includegraphics[scale=0.8]{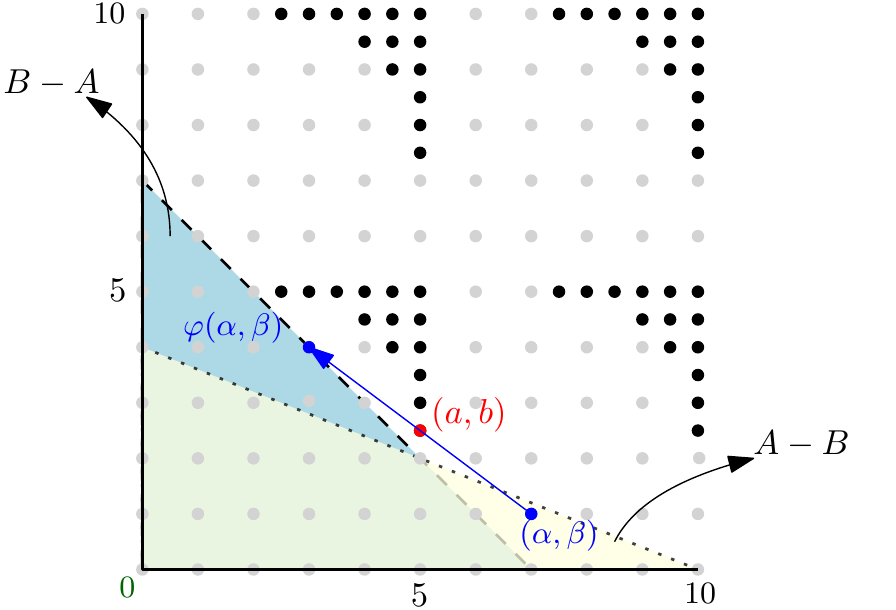}
    \caption{Figure illustrating the proof of Theorem \ref{th:wrmdimpar} for $d=7$, $q=11$, $(a,b)=(5,2.5)$, $A=\left\{ (i,j)\in \mathbb N^2 \mid 0.4 i + j \leq 4 \right\}$ and $B=\left\{ (i,j) \in \mathbb N^2 \mid i+j \leq 7\right\}$.}
    \label{fig:thwrmdimpar}
\end{figure}

Since the injectivity of $\varphi$ is easy to check, we are  proving that $\varphi$ is well defined in three steps:
\begin{itemize}
\item[(a)] if $(\alpha,\beta) \in A$, then $(2a-\alpha,2b-\beta) \notin A$,
\item[(b)] if $(\alpha,\beta) \in A - B$, then $(2a-\alpha,2b-\beta) \in \N^2$, and
\item[(c)] if $(\alpha,\beta) \in A - B$, then $(2a-\alpha,2b-\beta) \in B$.
\end{itemize}
If (a) does not hold, then both $(\alpha,\beta)$ and $(2a-\alpha,2b-\beta) \in A$. Hence, $(2a,2b) = (\alpha,\beta) + (2a-\alpha,2b-\beta) \in A + A$ and $\mC_A^{\, (2)} = \mC_{A+A}$. Since $\mC_A$ is a weighted Reed-Muller code, by Lemma \ref{Lemma:WRM-sharp} we have that $d \leq d(\mC^{\,(2)}) = {\rm FB}(\mC^{\,(2)}) \leq (q-2a)(q-2b) = d-1$, a contradiction.

We observe that $(2a-\alpha,2b-\beta) \in \Z^2$ and that $\alpha \leq q-1 = 2a$, so to prove (b) we just need to see that $2b - \beta \geq 0$. 
Assume that $2b < \beta$ and let us prove that 

\begin{itemize}
\item[(b.1)] $P_1 = (a,b + \frac{1}{2}),\ Q_1 = (a, b - \frac{1}{2}) \in A$ if $q$ is odd, or
\item[(b.2)] $P_2 = (a + \frac{1}{2}, b),\ Q_2 = (a - \frac{1}{2}, b)  \in A$ if $q$ is even. 
\end{itemize}

If $\alpha > a$, then $\alpha \geq a+\frac{1}{2}$ since $\beta \geq 2b + 1 > b + \frac{1}{2}$ we have that $P_1, Q_1 \in A$ in case (b.1) and $P_2, Q_2 \in A$ in case (b.2). If $\alpha \leq a$, from one side we have that $(\alpha, \beta) \notin B$, so 
\begin{equation}\label{eq:11}
\alpha + \beta \geq s + 1 = a+b+\frac{1}{2}
\end{equation} and, from the other side we have that $(\alpha,\beta) \in A$, which implies that 
\begin{equation}\label{eq:21}
w_1 \alpha + \beta \leq \lambda.
\end{equation} From (\ref{eq:11}) and (\ref{eq:21}) we get that
\[ (w_1 - 1) \alpha  + a + b + \frac{1}{2} \leq (w_1 - 1) \alpha + \alpha + \beta \leq w_1 \alpha + \beta \leq \lambda \leq a + b+ \frac{1}{2}\]
and, thus, $w_1 \leq 1$. Hence, using that $\alpha \leq a$, (\ref{eq:11}) and (\ref{eq:21}) we get that

\begin{eqnarray*}
w_1(a+\tfrac{1}{2}) + b & \leq  & w_1 a + b + \tfrac{1}{2} \\ 
&  = & a + b + \tfrac{1}{2} + (w_1-1)a \\ 
& \leq & a + b + \tfrac{1}{2} + (w_1-1)\alpha \\ 
& \leq & \alpha + \beta + (w_1-1)\alpha = \\ & = & w_1 \alpha + \beta \leq \lambda 
\end{eqnarray*}
and we conclude that $P_1,Q_1 \in A$ in case (b.1) and $P_2,Q_2 \in A$ in case  (b.2). Moreover, since $P_1 + Q_1 = P_2 + Q_2 = (2a,2b)$, in both cases we obtain that $(2a,2b) \in A+A$ and $\mC_A^{\, (2)} = \mC_{A+A}$. Since $\mC_A$ is a weighted Reed-Muller code, by Lemma \ref{Lemma:WRM-sharp} we have that $d \leq d(\mC^{\, (2)}) = {\rm FB}(\mC^{\,(2)}) \leq (q-2a)(q-2b) = d-1$, a contradiction.

Let us prove now (c). Whenever $(\alpha,\beta) \in A-B$, then $\alpha + \beta \geq s + 1$. Since $a+b = s+\frac{1}{2}$, we have that $2a-\alpha  + 2b-\beta \leq s$ and $(2a-\alpha, 2b-\beta) \in \N^2$ by (b), so $(2a-\alpha,2b-\beta) \in B$.

\textbf{Case II:} $\lambda > a + b + \frac{1}{2}$. We claim that $\frac{\lambda}{w_1} < a + b + \frac{1}{2}$. Otherwise, we have that  
$(a + \frac{1}{2},b), (a - \frac{1}{2},b) \in A$ if $q$ is even, or
$(a,b + \frac{1}{2}), (a,b - \frac{1}{2}) \in A$ if $q$ is odd. In both cases $(2a,2b) \in A + A$ and $\mC_A^{\, (2)} = \mC_{A+A}$. Since $\mC_A$ is a weighted Reed-Muller code, by Lemma \ref{Lemma:WRM-sharp} we have that $d \leq d(\mC^{\, (2)}) \leq (q-2a)(q-2b) = d-1$, a contradiction. 

Since $\frac{\lambda}{w_1} < a + b + \frac{1}{2}$, then $A = \{(i,j) \in \N^2 \, \vert \, 0 \leq i,j \leq q-1 {\text \ and \ } i + \frac{1}{w_1}{j} \leq \frac{\lambda}{w_1}\}$ and a symmetric argument to {\it Case I} applies here. 
\end{proof}

Since $(\mathrm{RM}_q(2, s)^{\, (2)}) = d$, this means that $\mathrm{RM}_q(2, s)$ has the highest dimension among all the weighted Reed-Muller codes $\mC$ such that $d(\mC^{\, (2)}) \geq d$.

 \begin{lemma} \label{lemma:WRM} Let $\F_q$ be a finite field and $d \in \Z^+$ be an even integer with $d < q$ and let $s := q-\frac{d}{2}$. Let \[ \begin{array}{lll} B_1  & := & \{(i,j) \in \N^2 \, \vert \, i+j < s\} \cup 
 \{(i,j) \in \N^2 \, \vert \, i + j = s {\text \ and\ } j < (q-d+1)/2\}, {\text \ and } \\ B_2 & := & \{(i,j) \in \N^2 \, \vert \, i+j < s\} \cup 
 \{(i,j) \in \N^2 \, \vert \, i + j = s {\text \ and\ } i < (q-d+1)/2\} \end{array}\] then, $\mC_{B_1}$ and $\mC_{B_2}$ are weighted Reed-Muller codes and $k(\mC_{B_1}) = k(\mC_{B_2})$.
 \end{lemma}

 \begin{proof}
 It suffices to perturb slightly the line $x + y = s$ to get that both $\mC_{B_1}$ and $\mC_{B_2}$ are weighted Reed-Muller codes. Indeed it is easy to check that 
 $$k(\mC_{B_1}) = |B_1| = \frac{(q-\frac{d}{2}+2)(q-\frac{d}{2}+1)}{2} + \frac{q-\frac{d}{2}}{2} -1 =  k(\mC_{B_2})
 $$
 See an illustration in Figure \ref{fig:lemma-WRM}.
 \end{proof}

 \begin{figure}[h!]
    \centering
    \begin{subfigure}[b]{0.45\textwidth}
        \includegraphics[scale=0.8]{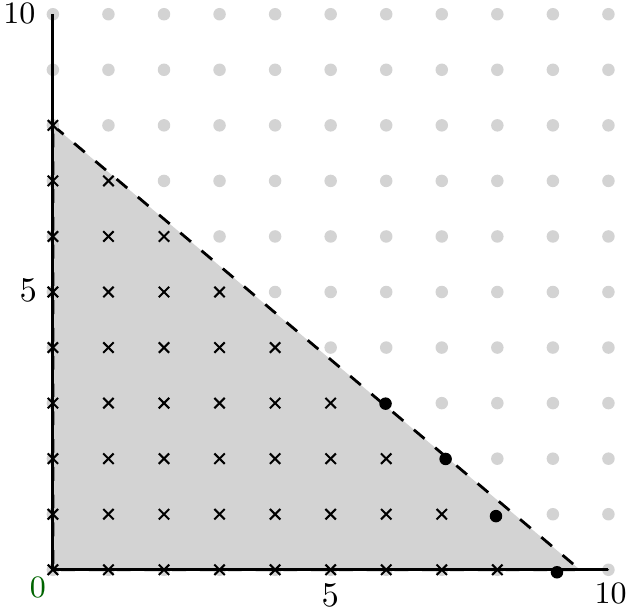}
        \caption{Example of the code $\mC_{B_1}$ of Lemma \ref{lemma:WRM}.}
    \end{subfigure}
    \begin{subfigure}[b]{0.45\textwidth}
        \includegraphics[scale=0.8]{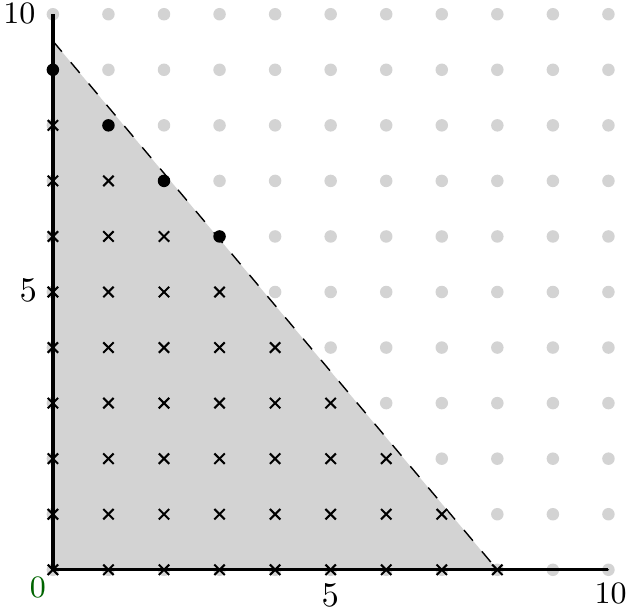}
        \caption{Example of the code $\mC_{B_2}$ of Lemma \ref{lemma:WRM}.}
    \end{subfigure}\\
    \caption{Example illustrating Lemma \ref{lemma:WRM} with $q=11$ and $d=4$.}
    \label{fig:lemma-WRM}
\end{figure}

We can now consider the case when the minimum distance is even.

 \begin{theorem}\label{th:wrmdpar} Let $\F_q$ be a finite field and $d \in \Z^+$ be an even integer with $d < q$. If $\mC$ is a weighted Reed-Muller code over $\F_q$ with $d(\mC^{\, (2)}) \geq d$, then
$k(\mC) \leq k(\mC_B),$ where $\mC_B$ is any of the weighted Reed-Muller codes described in Lemma \ref{lemma:WRM}.
 \end{theorem}
\begin{proof}This proof will follow the same ideas in Theorem \ref{th:wrmdimpar}. 
Let $\mC$ be a weighted Reed-Muller code over $\F_q$ with $d(\mC^{\, (2)}) \geq d$. We assume without loss of generality that $\mC = \mathrm{WRM}_q(\lambda,2,w_1,1)$ for some $\lambda, w_1 > 0$. Taking 
$$A := \{(i,j) \in \lii 0, q-1 \rii \, \vert \,  w_1 i + j \leq \lambda\}$$ we have that $\mC = \mC_A$. 

In this proof we denote $a:= (q-1)/2$ and $b := (q-d+1)/2$; and observe that either $(a,b) \in \N^2$ or both $(a-\frac{1}{2}, b + \frac{1}{2}), (a+\frac{1}{2}, b - \frac{1}{2}) \in\N^2$.  We divide the proof in two cases depending on the value of $\lambda$.

\textbf{Case I:} $\lambda \leq a + b$. We take  $B = B_1$ as in Lemma \ref{lemma:WRM}. To prove that $|A| = k(\mC) \leq k(\mC_B) = |B|$ we are going to prove that either $A \subseteq B$, or  
\[ \begin{array}{lccl} \varphi: & A - B &\ \longrightarrow\ & B - A \\ & (\alpha,\beta) & \mapsto & (2a-\alpha, 2b - \beta) \end{array}\]
is an injective map (see Figure \ref{fig:thwrmdpar} for a graphic representation of this idea). 

 \begin{figure}[h!]
    \centering
        \includegraphics[scale=0.8]{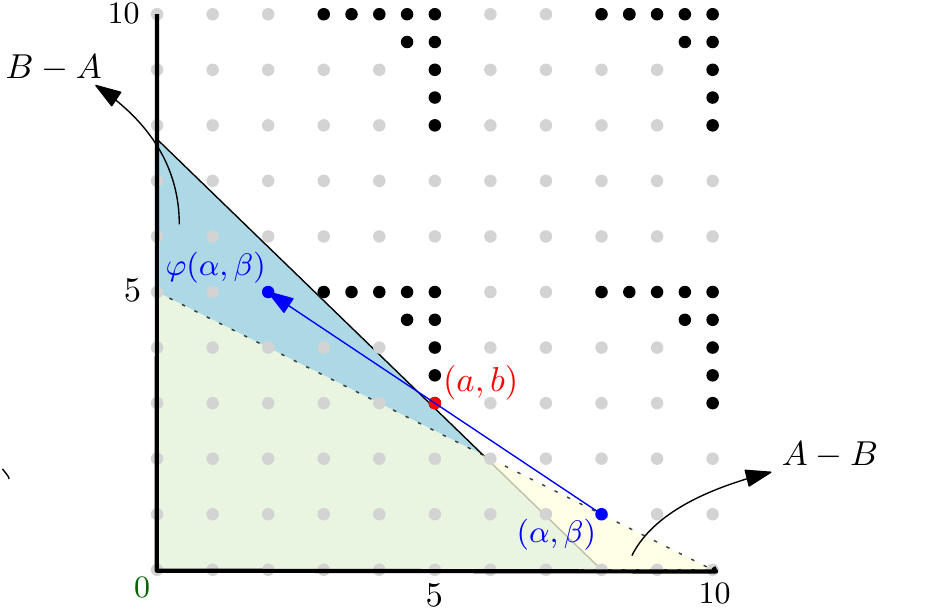}
    \caption{Figure illustrating the proof of Theorem \ref{th:wrmdpar} for $d=6$, $q=11$, $(a,b)=(5,3)$, $A=\left\{ (i,j)\in \mathbb N^2 \mid 0.4 i + j \leq 4 \right\}$ and $B=\left\{ (i,j) \in \mathbb N^2 \mid i+j < 8\right\} \cup \left\{ (i,j) \in \mathbb N^2 \mid i+j = 8 \hbox{ and } j<3\right\}$.}
    \label{fig:thwrmdpar}
\end{figure}

Since the injectivity of $\varphi$ is easy to check, we are showing that $\varphi$ is well defined in three steps:
\begin{itemize}
\item[(a)] if $(\alpha,\beta) \in A$, then $(2a-\alpha,2b-\beta) \notin A$,
\item[(b)] if $(\alpha,\beta) \in A - B$, then $(2a-\alpha,2b-\beta) \in \N^2$, and
\item[(c)] if $(\alpha,\beta) \in A - B$, then $(2a-\alpha,2b-\beta) \in B$.
\end{itemize}
If (a) does not hold, then both $(\alpha,\beta)$ and $(2a-\alpha,2b-\beta) \in A$. Hence, $(2a,2b) = (\alpha,\beta) + (2a-\alpha,2b-\beta) \in A + A$ and $\mC_A^{\, (2)} = \mC_{A+A}$. Since $\mC_A$ is a weighted Reed-Muller code, by Lemma \ref{Lemma:WRM-sharp} we have that 
$$d \leq d(\mC^{\, (2)}) = \mathrm{FB}(\mC^{\, (2)}) \leq (q-2a)(q-2b) = d-1,$$ a contradiction.

We observe that $(2a-\alpha,2b-\beta) \in \Z^2$ and that $\alpha \leq q-1 = 2a$, so to prove (b) we just need to see that $2b - \beta \geq 0$. 
Assume that $2b < \beta$ and let us prove that 

\begin{itemize}
\item[(b.1)] $P = (a,b) \in A$ if $q$ is odd, or
\item[(b.2)] $Q_1 = (a - \frac{1}{2},  b + \frac{1}{2}),\ Q_2 = (a + \frac{1}{2}, b- \frac{1}{2})  \in A$ if $q$ is even. 
\end{itemize}

If $\alpha > a$, then $\alpha \geq a+\frac{1}{2}$ since $\beta \geq 2b + 1 > b + \frac{1}{2}$ we have that $P \in A$ in case (b.1) and $Q_1, Q_2 \in A$ in case (b.2). If $\alpha \leq a$, from one side we have that $(\alpha, \beta) \notin B$, so 
\begin{equation}\label{eq:1}
\alpha + \beta \geq a+b
\end{equation} and, if we have equality, then $\beta \geq b$. From the other side we have that $(\alpha,\beta) \in A$, which implies that 
\begin{equation}\label{eq:2}
w_1 \alpha + \beta \leq \lambda.
\end{equation} From (\ref{eq:1}) and (\ref{eq:2}) we get that
\[ (w_1 - 1) \alpha  + a + b  \leq (w_1 - 1) \alpha + \alpha + \beta = w_1 \alpha + \beta \leq \lambda \leq a + b\]
and, thus, $w_1 \leq 1$. Hence, we separate three cases:

\textbf{Subcase I.I.} If $\alpha + \beta > a + b$. 

\begin{eqnarray*} 
w_1(a+\tfrac{1}{2}) + b - \tfrac{1}{2} & \leq  & w_1 a + b  \leq w_1(a-\tfrac{1}{2}) + b + \tfrac{1}{2}  \\ &  < & w_1 a + b + \tfrac{1}{2}  \\
& = &  a + b + \tfrac{1}{2} + (w_1-1)a  \\ & \leq & a + b + \tfrac{1}{2} + (w_1-1)\alpha  \\ & < & \alpha + \beta + (w_1-1)\alpha \\ & = & w_1 \alpha + \beta \leq \lambda. 
\end{eqnarray*}

So, $P \in A$ if $q$ is odd, or both $Q_1,Q_2 \in A$ if $q$ is even. 

\textbf{Subcase I.II.} If $\alpha + \beta = a + b$ and $q$ is odd. 
Since $\beta \geq b$ and $w_1 < 1$, we have that $w_1(\alpha - a) + \beta - b \geq w_1 (\alpha - a + \beta - b) = 0$. As a consequence,

\[ w_1 a + b \leq w_1 a + b + w_1(\alpha - a) + \beta - b  = w_1 \alpha +  \beta \leq \lambda.\]Therefore $P \in A$.

\textbf{Subcase I.III.} If $\alpha + \beta = a + b$ and $q$ is even.
Since $\beta \geq b$ and $b \notin \N$, then $\beta \geq b + \frac{1}{2}$; moreover, $w_1 < 1$, then we have that 
$w_1(\alpha - a + \frac{1}{2}) + \beta - b  - \frac{1}{2} \geq w_1 (\alpha - a + \frac{1}{2} + \beta - b - \frac{1}{2}) = 0$. As a consequence,

\begin{eqnarray*}
w_1(a+\tfrac{1}{2}) + b - \tfrac{1}{2} & \leq  & w_1(a-\tfrac{1}{2}) + b + \tfrac{1}{2}  \\ &  \leq &   w_1(a-\tfrac{1}{2}) + b + \tfrac{1}{2} +  w_1(\alpha - a + \tfrac{1}{2}) + \beta - b  - \tfrac{1}{2} \\ & = & w_1 \alpha + \beta \leq \lambda 
\end{eqnarray*}
and we conclude that $Q_1,Q_2 \in A$.

Moreover, since $P  + P = Q_1 + Q_2 = (2a,2b)$, in both cases we obtain that $(2a,2b) \in A+A$ and $\mC_A^{\, (2)} = \mC_{A+A}$. Since $\mC_A$ is a weighted Reed-Muller code, by Lemma \ref{Lemma:WRM-sharp} we have that $d \leq d(\mC^{\, (2)}) \leq (q-2a)(q-2b) = d-1$, a contradiction.

Let us prove now (c). Take $(\alpha,\beta) \in A-B$, then either 

\begin{itemize}
\item[(c.1)] $\alpha + \beta > a+b$, or
\item[(c.2)] $\alpha + \beta = a+b$ and $\beta \geq b$. 
\end{itemize}
In (c.1) we have that $2a-\alpha  + 2b-\beta < a+b$, so $(2a-\alpha,2b-\beta) \in B$.
In (c.2) we observe that $\beta \neq b$ because $(a,b) \notin A$. Then, we have that $2a-\alpha  + 2b-\beta = a+b$ and $2b - \beta < b$, so $(2a-\alpha,2b-\beta) \in B$.

\textbf{Case II:} $\lambda \geq a + b$. We claim that $\frac{\lambda}{w_1} < a + b$. Otherwise, we have that  
$P \in A$ if $q$ is odd, or
$Q_1,Q_2 \in A$ if $q$ is even. In both cases $(2a,2b) \in A + A$  and $\mC_A^{\, (2)} = \mC_{A+A}$. Since $\mC_A$ is a weighted Reed-Muller code, by Lemma \ref{Lemma:WRM-sharp} we have that $d \leq d(\mC^{\, (2)}) \leq (q-2a)(q-2b) = d-1$, a contradiction. Since $\frac{\lambda}{w_1} < a + b$, then $A = \{(i,j) \in \N^2 \, \vert \, 0 \leq i,j \leq q-1 {\text \ and \ } i + \frac{1}{w_1}{j} \leq \frac{\lambda}{w_1}\}$ and a symmetric argument to \textbf{Case I} using $B = B_2$ with $B_2$ as in Lemma \ref{lemma:WRM} applies here.
\end{proof}

Finally, we are proving that when $d$ is small enough (more precisely, when $d < (2-\sqrt{2})q$), then there are weighted Reed-Muller codes that have more dimension and whose square has the same designed minimum distance as the corresponding half hyperbolic code.

\begin{proposition}\label{prop:odd} If $d < (2-\sqrt{2})q$ and $d$ odd, then 
$$k(\mathrm{RM}_q\left(2,q-\frac{d-1}{2}\right)) > k(\mathrm{HalfHyp}_q(2,d)).$$
\end{proposition}

\begin{proof}
Take notice that
\begin{eqnarray*}
k(\mathrm{RM}_q\left(2,q-\frac{d-1}{2} \right))  & = & \frac{(q-\frac{d-1}{2}+2) (q-\frac{d-1}{2}+1)}{2} \\
 & = & \frac{(2q-d+5)(2q-d+3)}{8} = A
\end{eqnarray*}

\begin{eqnarray*}
k(\mathrm{HalfHyp}_q(2,d))  & < & \left( \frac{q+1}{2}\right)^2 - 2\left( \frac{d-1}{2}\right)+1  = B
\end{eqnarray*}

Therefore if $A-B > 0$ then our claim holds.
Now $A-B > 0$ if $p(d)= d^2-4qd + (2q^2+12q+13) >0$. This defines a quadratic function whose vertex represent its minimum value. That is, $p(d)>0$ if $d>2q + \sqrt{2q^2-12q-13}$ or $d<2q - \sqrt{2q^2-12q-13}$. Take notice that if 
$$d<(2-\sqrt{2})q<2q - \sqrt{2q^2-12q-13}$$ then:
$k(\mathrm{RM}_q\left(2,q-\frac{d-1}{2}\right)) > k(\mathrm{HalfHyp}_q(2,d))$.
\end{proof}

\begin{proposition}\label{prop:even}
If $d<(2-\sqrt{2})q$ and $d$ even, then $k(\mC_B) > k(\mathrm{HalfHyp}_q(2,d))$ where $\mC_B$ is one of the weighted Reed-Muller codes defined in Lemma \ref{lemma:WRM}.
\end{proposition}

\begin{proof}
Take notice that
\begin{eqnarray*}
k(\mC_B)  & = & \frac{(q-\frac{d-1}{2}+2) (q-\frac{d-1}{2}+1)}{2}  + \frac{q-\frac{d}{2}}{2}-1\\
 & = & \frac{1}{8}d^2 - \frac{1}{2}dq + \frac{1}{2}q^2 - d + 2q = A
\end{eqnarray*}

\begin{eqnarray*}
k(\mathrm{HalfHyp}_q(2,d))  & < & \left( \frac{q+1}{2}\right)^2 - 2\left( \frac{d-1}{2}\right)+1  = B
\end{eqnarray*}

Therefore if $A-B > 0$ then our claim is true.
Now $A-B > 0$ if $p(d)= d^2 - 4dq + (2q^2 + 12q - 18) >0$. This defines a quadratic function whose vertex represent its minimum value. That is, $p(d)>0$ if $d>2q + \sqrt{2q^2-12q+18}$ or $d<2q - \sqrt{2q^2-12q+18}$. 
Take notice that if 
$$d<(2-\sqrt{2})q<2q - \sqrt{2q^2-12q+18},$$ then:
$k(\mC_B) > k(\mathrm{HalfHyp}(2,d))$.
\end{proof}

\begin{example}
\label{FinalExample}
We continue with Example \ref{example:HHYP}.
That is, consider $\mC_B = \mathrm{Hyp}_{11}(6,2)$ with 
$B=\left\{ (i,j)\in \lii 0,10\rii^2 \mid (11-i) (11-j) \geq 6\right\}$.
We recall that the half hyperbolic code of order $6$ has parameters $[11^2, 25,49]_{11}$ and $d(\mC_A^{\, (2)})\geq 6$.

Taking  $A_1$ as in Figure \ref{FinalExample:1}, then $\mC_{A_1}$ is a weighted Reed-Muller code with parameters $[11^2, 34, 9]_{11}$ such that $d(\mC_{A_1}^{\, (2)})\geq 6$. Take notice that this example already gives an affine variety code with higher dimension than the half hyperbolic code.

Moreover, by Theorem \ref{th:wrmdpar}, we know that if we take the set $A_2$ defined in Figure \ref{FinalExample:2} then the weighted Reed-Muller code $\mC_{A_2}$ has higher dimension than any other weighted Reed-Muller code $\mC$ such that $d(\mC^{\, (2)})\geq 6$. In particular, we know by Theorem \ref{th:wrmdpar} that $k(\mC_{A_2})\geq k(\mC_{A_1})$. Note that $\mC_{A_2}$ is defined as in Lemma \ref{lemma:WRM} and has parameters $[11^2, 39, 33]_{11}$.
  
\begin{figure}[h!]
    \centering
    \begin{subfigure}[b]{0.4\textwidth}
        \includegraphics[scale=0.8]{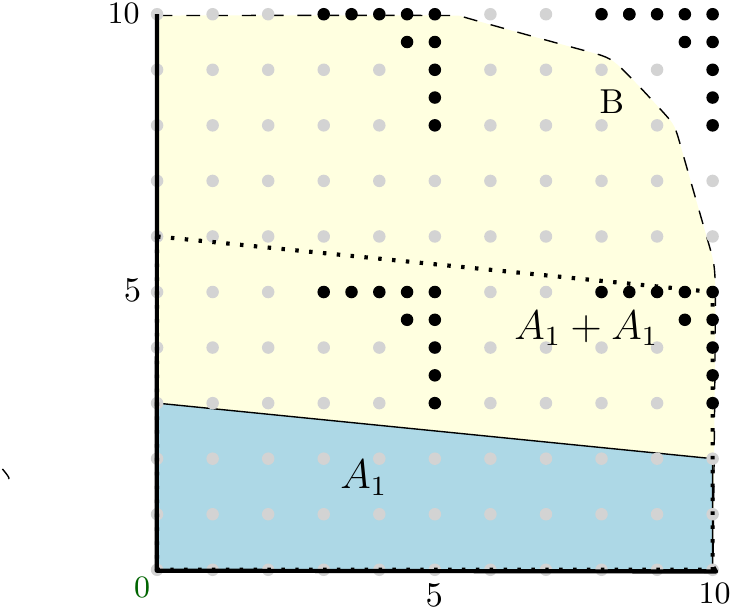}
        \caption{Set $A_1$ of Example \ref{FinalExample}.}
        \label{FinalExample:1}
    \end{subfigure}
   \hfill
    ~ 
    \begin{subfigure}[b]{0.55\textwidth}
        \includegraphics[scale=0.8]{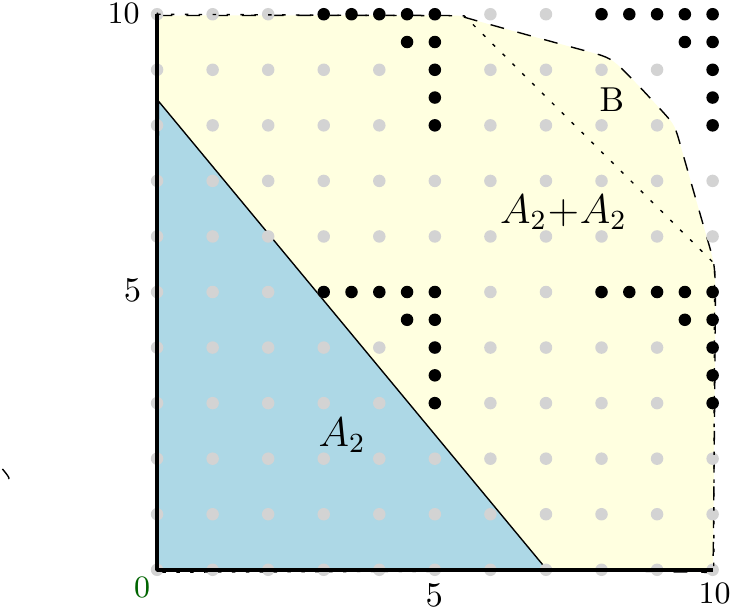}
        \caption{Set $A_2$ of Example \ref{FinalExample}.}
        \label{FinalExample:2}
    \end{subfigure}
    \caption{Figure illustrating Example \ref{FinalExample}.}
    \label{fig:examples}
\end{figure}
\end{example}




\appendix
\section{For which affine codes $\mathcal C_A$ is it verified that $\mathrm{FB}(\mathcal C_A) = d(\mathcal C_A)$?}

Let $A\subseteq \lii 0,q-1\rii^m$ and consider the code $\mathcal C_A$ as the affine variety code $C(I,L)$ with $I=(0)$ and $L=\mathbb F_q[A]$. Then, we know that the length of $\mathcal C_A$ is $q^m$ and its dimension coincides with the cardinality of the set $A$. Moreover its minimum distance, denoted as $d(\mathcal C_A)$, satisfies that $d(\mathcal C_A)\geq \mathrm{FB}(\mathcal C_A)$. In this section we will study when these two values coincide. More concretely, we provide sufficient conditions to have the equality  $d(\mathcal C_A) = \mathrm{FB}(\mathcal C_A)$.

\begin{lemma}
\label{FB-1}
Suppose that $\mathrm{FB}(\mathcal C_A) = (q-\alpha_1)\cdots (q-\alpha_m)$. Then $d(\mathcal C_A) = \mathrm{FB}(\mathcal C_A)$ if all the elements $\beta = (\beta_1, \ldots, \beta_m)$ with $0\leq \beta_i \leq \alpha_i$ belong to the set $A$.
\end{lemma}

\begin{proof}
First, to simplify the proof let us suppose that $m=2$. Let $\mathcal P = \left\{ P_1, \ldots, P_n\right\}$ be the ordered enumeration of the $q^2$ different points of $\mathbb F_q^2$. Suppose that $\mathrm{FB}(\mathcal C_A) = (q-\alpha_1)(q-\alpha_2)$. Now we can define the polynomial
$$f(x) = (X_1-P_1) \cdots (X_1-P_{\alpha_1}) \cdot (X_2-P_1) \cdots (X_2-P_{\alpha_2}).$$
Take notice that by hypothesis $f(X_1, X_2) \in \mathbb F_q[A]$ since all the elements $\beta = (\beta_1, \beta_2)$ with $0\leq \beta_1 \leq \alpha_1$ and $0\leq \beta_2 \leq \alpha_2$ belongs to the set $A$. 
Moreover, the $\mathbb F_q$-roots of $f$ are all the points of form:
$$\begin{array}{cccc}
(P_i, z_2) & \hbox{ and } & (z_1, P_j) & \hbox{ with } i \in \{1, \ldots, \alpha_1\} \hbox{ , } j \in \{1, \ldots, \alpha_2\} \hbox{ and } z_1, z_2\in \mathbb F_q.
\end{array}$$
That is, the number of $\mathbb F_q$-roots of $f(x)$ is $(\alpha_1 + \alpha_2) q - \alpha_1 \alpha_2$. Therefore, we have found a codeword $\mathbf c = \mathrm{ev}_{\mathcal P} (f)\in \mathcal C_A$ of weight $q^2 - (\alpha_1 + \alpha_2) q - \alpha_1 \alpha_2 = \mathrm{FB}(\mathcal C_A)$. Hence the minimum distance of $\mathcal C_A$ is $ \mathrm{FB}(\mathcal C_A)$.

The generalization to $m$ variables is straightforward. 
Let $\mathcal P = \left\{ P_1, \ldots, P_n\right\}$ be the ordered enumeration of the $q^m$ different points of $\mathbb F_q^m$.
Then, using all the hypothesis we can define the following polynomial in $\mathbb F_q[A]$:
$$f(X_1, \ldots, X_m) = \prod_{i=1}^m (X_i - P_1) \cdots (X_i - P_{\alpha_i}) \in \mathbb F_q[A].$$
Thus, we have found a codeword of $\mathcal C_A$ of weight $\mathrm{FB}(\mathcal C_A)$, hence $d(\mathcal C_A) = \mathrm{FB}(\mathcal C_A)$.
\end{proof}

The following result shows that if $l$ is a divisor of $q-1$ then, there exists a polynomial $f(x) = X^l-\alpha\in \mathbb F_q[X]$ with small support but a large number of $\mathbb F_q$-roots. This result will be useful for computing the minimum distance of codes of type $\mathcal C_A$ by just checking that a very small number of points belongs to the set $A$.

\begin{lemma}
\label{FB-2}
Let $\alpha$ be a primitive element of $\mathbb F_q^*$. Consider the polynomial $f(X) = X^l-\alpha^j \in \mathbb F_q[X]$. Then $X^l-\alpha^j$ has at least one roots in $\mathbb F_q$ if and only if $\gcd(l, q-1)$ divides $j$. 
In such case, the exactly number of $\mathbb F_q$-roots of $f(X)$ is $\gcd(l,q-1)$.
\end{lemma}

\begin{proof}
Suppose that $\alpha^i$ is an $\mathbb F_q$-root of $f(X)$, then $f(\alpha^i) = 0$, that is $\alpha^{il} = \alpha^j$ which implies that the order of $\alpha$, which is $q-1$, divides $il-1$. In other words, there exists an integer $x$ such that $x(q-1) + il = j$. Take notice that such $x$ exists if and only if $\gcd (l,q-1)$ divides $j$.

In such case, if $(x,y)$ is a solution of the equation $x(q-1) + yl = j$. Then, all solutions of this equations has the form:
$$\left( x-\lambda \frac{l}{\gcd (l,q-1)}, y + \lambda \frac{q-1}{\gcd(l,q-1)} \right) \hbox{ with } \lambda \in \mathbb Z$$
Therefore, if $f(X)$ has at least one root in $\mathbb F_q$, then it will have exactly $\gcd(l,q-1)$ $\mathbb F_q$-roots.
\end{proof}

\begin{corollary}
Suppose that $\mathrm{FB}(\mathcal C_A) = (q-l)q^{m-1}$ with $l$ a divisor of $q-1$. Then, $d(\mathcal C_A) = \mathrm{FB}(\mathcal C_A)$ if $\{1, X_i^l\}\subseteq \mathbb F_q[A]$ for some $i\in \{1, \ldots, m\}$.
\end{corollary}

\begin{proof}
By hypothesis we can define the following polynomial in $\mathbb F_q[A]$:
$$f(X) = X_i^l - \beta \hbox{ for certain }\beta \in \mathbb F_q.$$
Then, by Lemma \ref{FB-2}, $f(X)$ has $lq^{m-1}$ $\mathbb F_q$-roots. That is, we have found a codeword of $\mathcal C_A$ of weight $\mathrm{FB}(\mathcal C_A)$, hence $d(\mathcal C_A) = \mathrm{FB}(\mathcal C_A)$.
\end{proof}

\begin{lemma}
Suppose that $\mathrm{FB}(\mathcal C_A) = (q-kl)q^{m-1}$ with $l$ a divisor of $q-1$. Then, $d(\mathcal C_A) = \mathrm{FB}(\mathcal C_A)$ if $\{1, X_i^l, X_i^{2l}, \ldots, X_i^{kl}\}\subseteq \mathbb F_q[A]$ for some $i\in \{1, \ldots, m\}$.
\end{lemma}

\begin{proof}
By hypothesis we can define the following polynomial in $\mathbb F_q[A]$:
$$f(X) = (X_i^l - \beta) (X_i^{2l}-\beta^2) \cdots (X_i^{kl} - \beta^k) \hbox{ for certain }\beta \in \mathbb F_q.$$
Then, by Lemma \ref{FB-2}, $f(X)$ has $klq^{m-1}$ $\mathbb F_q$-roots. That is, we have found a codeword of $\mathcal C_A$ of weight $\mathrm{FB}(\mathcal C_A)$, hence $d(\mathcal C_A) = \mathrm{FB}(\mathcal C_A)$.
\end{proof}

\begin{lemma}
Suppose that $\mathrm{FB}(\mathcal C_A) = (q-l_1)\cdots (q-l_m)$ with $l_i$ a divisor of $q-1$. Then, $d(\mathcal C_A) = \mathrm{FB}(\mathcal C_A)$ if $\{1, X_1^{l_1}, \cdots, X_m^{l_m}\}\subseteq \mathbb F_q[A]$.
\end{lemma}

\begin{proof}
By hypothesis we can define the following polynomial in $\mathbb F_q[A]$:
$$f(X) = \prod_{i=1}^m( X_i^{l_i} - \beta_i) \hbox{ for certain }\beta_1, \ldots, \beta_m \in \mathbb F_q.$$
Then, by Lemma \ref{FB-2}, $f(X)$ has 
$$(l_1 + \ldots l_m) q^{m-1} - \sum_{1\leq i < j\leq m} l_i l_j q^{m-2} - \sum_{1\leq i < j<k\leq m} l_i l_jl_k q^{m-3}
- \ldots - 
l_1 \cdots l_m
$$ 
$\mathbb F_q$-roots. That is, we have found a codeword of $\mathcal C_A$ of weight $\mathrm{FB}(\mathcal C_A)$, hence $d(\mathcal C_A) = \mathrm{FB}(\mathcal C_A)$.
\end{proof}

\begin{lemma}
Suppose that $\mathrm{FB}(\mathcal C_A) = (q-k_1l_1)\cdots (q-k_ml_m)$ with $l_i$ a divisor of $q-1$. Then, $d(\mathcal C_A) = \mathrm{FB}(\mathcal C_A)$ if $\{1, X_1^{l_1}, \ldots, X_1^{ml_1}, \cdots, X_m^{l_m}, \ldots, X_m^{k_ml_m}\}\subseteq \mathbb F_q[A]$.
\end{lemma}

\begin{proof}
By hypothesis we can define the following polynomial in $\mathbb F_q[A]$:
$$f(X) = \prod_{i=1}^m( X_i^{l_i} - \beta_i) (X_i^{2l_i}-\beta_i^2) \cdots (X_i^{k_il_i} - \beta_i^{k_i}) \hbox{ for certain }\beta_1, \ldots, \beta_m \in \mathbb F_q.$$
Then, by Lemma \ref{FB-2}, we have found a codeword of $\mathcal C_A$ of weight $\mathrm{FB}(\mathcal C_A)$, hence $d(\mathcal C_A) = \mathrm{FB}(\mathcal C_A)$.
\end{proof}

\begin{lemma}
Suppose that $\mathrm{FB}(\mathcal C_A) = (q-l)q^{m-1}$ with $l-1$ a divisor of $q-1$. Then, $d(\mathcal C_A) = \mathrm{FB}(\mathcal C_A)$ if $\{X_i, X_i^l\}\subseteq \mathbb F_q[A]$ for some $i\in \{1, \ldots, m\}$.
\end{lemma}

\begin{proof}
By hypothesis we can define the following polynomial in $\mathbb F_q[A]$:
$$f(X) = (X_i^l - \beta X_i) = X_i (X_i^{l-1}-\beta) \hbox{ for certain }\beta \in \mathbb F_q.$$
Then, by Lemma \ref{FB-2}, $f(X)$ has 
$lq^{m-1}$ $\mathbb F_q$-roots. That is, we have found a codeword of $\mathcal C_A$ of weight $\mathrm{FB}(\mathcal C_A)$, hence $d(\mathcal C_A) = \mathrm{FB}(\mathcal C_A)$.
\end{proof}

\begin{lemma}
Let $A\subseteq \lii 0,q-1\rii^m$ and $s\in \lii 0,q-1\rii $. If for all $f\in \mathbb F_q[A]$ we have that $X_1^s$ is a divisor of $f(X)$, then $d(\mathcal C_A) = d(\mathcal C_B)$ with $$B=\{ (i_1-s-1, i_2,\ldots, i_m)\mid (i_1, \ldots, i_m) \in A\}.$$
The result can be generalized to any other coordinate $X_i$ with $i=2, \ldots, m$.
\end{lemma}

\begin{proof}
By hypothesis every polynomial $f\in \mathbb F_q[A]$ can be written as $f=X_1^s g$ with $g\in \mathbb F_q[B]$. And both polynomials $f$ and $g$ have exactly the same number of $\mathbb F_q$-roots.
\end{proof}

\end{document}